\documentclass[reprint,superscriptaddress,amsmath,amssymb,aps,pra,floatfix]{revtex4-2}
\usepackage{graphicx} 
\usepackage{xcolor}
\usepackage{etoolbox}
\usepackage{physics}
\usepackage{mathtools}
\usepackage{bbm}
\usepackage{amsthm}
\usepackage{hyperref}
\hypersetup{
    colorlinks,
    citecolor=blue,
    filecolor=blue,
    linkcolor=blue,
    urlcolor=blue
}

\usepackage{cleveref}
\crefname{equation}{Eq.}{Eqs.}
\crefname{figure}{Fig.}{Figs.}
\crefname{theorem}{Theorem}{Theorems}
\crefname{corollary}{Corollary}{Corollaries}
\crefname{section}{Sec.}{Secs.}
\crefname{lemma}{Lemma}{Lemmas}

\makeatletter
\newtheorem*{rep@theorem}{\rep@title}
\newcommand{\newreptheorem}[2]{%
\newenvironment{rep#1}[1]{%
 \def\rep@title{#2 \ref{##1}}%
 \begin{rep@theorem}}%
 {\end{rep@theorem}}}
\makeatother

\makeatletter
\def\maketitle{
\@author@finish
\title@column\titleblock@produce
\suppressfloats[t]}
\makeatother

\newtheorem{theorem}{Theorem}
\newreptheorem{theorem}{Theorem}
\newtheorem{lemma}{Lemma}

\newtheorem{corollary}{Corollary}
\newreptheorem{corollary}{Corollary}

\newcommand{\ketW}[1][M]{\ket{\operatorname{W}_{#1}}}

\newcommand{\isection}[1]{%
\phantomsection\addcontentsline{toc}{section}{#1}\emph{#1---}%
}

\allowdisplaybreaks
\begin{document}

\phantomsection\addcontentsline{toc}{part}{Sufficient Wigner Negativity Implies Genuine Multipartite Entanglement}

\title{Sufficient Wigner Negativity Implies Genuine Multipartite Entanglement}

\author{Lin Htoo Zaw}
\affiliation{Centre for Quantum Technologies, National University of Singapore, 3 Science Drive 2, Singapore 117543}

\author{Jiajie Guo}
\affiliation{State Key Laboratory of Artificial Microstructure and Mesoscopic Physics, School of Physics, Frontiers Science Center for Nano-optoelectronics, $\&$ Collaborative Innovation Center of Quantum Matter, Peking University, Beijing 100871, China}

\author{Qiongyi He}
\email{qiongyihe@pku.edu.cn}
\affiliation{State Key Laboratory of Artificial Microstructure and Mesoscopic Physics, School of Physics, Frontiers Science Center for Nano-optoelectronics, $\&$ Collaborative Innovation Center of Quantum Matter, Peking University, Beijing 100871, China}
\affiliation{Collaborative Innovation Center of Extreme Optics, Shanxi University, Taiyuan, Shanxi 030006, China}
\affiliation{Hefei National Laboratory, Hefei 230088, China}

\author{Matteo Fadel}
\email{fadelm@phys.ethz.ch}
\affiliation{Department of Physics, ETH Z\"{u}rich, 8093 Z\"{u}rich, Switzerland}

\author{Shuheng Liu}
\email{liushuheng@pku.edu.cn}
\affiliation{State Key Laboratory of Artificial Microstructure and Mesoscopic Physics, School of Physics, Frontiers Science Center for Nano-optoelectronics, $\&$ Collaborative Innovation Center of Quantum Matter, Peking University, Beijing 100871, China}

\begin{abstract}
Wigner negativity and genuine multipartite entanglement (GME) are key nonclassical resources that enable computational advantages and broader quantum-information tasks.
In this work, we prove two theorems for multimode continuous-variable systems that relate these nonclassical resources.
Both theorems show that sufficient Wigner negativity---either a sufficiently-large Wigner negativity volume along a suitably-chosen two-dimensional slice, or a sufficiently-large nonclassicality depth of the center-of-mass mode of a system---certifies the presence of GME.
Moreover, violations of the latter inequality provide lower bounds of the trace distance to the set of non-GME states.
Our results also provide sufficient conditions for generating GME by interfering a state with the vacuum through a multiport interferometer, complementing long-known necessary conditions.
Beyond these fundamental connections, our methods have practical advantages for systems with native phase-space measurements: they require only measuring the Wigner function over a finite region, or measuring a finite number of characteristic function points.
Such measurements are frequently performed with readouts common in circuit and cavity quantum electrodynamic systems, trapped ions and atoms, and circuit quantum acoustodynamic systems.
As such, our GME criteria are readily implementable in these platforms.
\end{abstract}

\maketitle

\isection{Introduction}%
Quantum theory predicts many fundamental effects that contradict our expectations from classical physics.
In particular, it rules out the assignment of classical joint probabilities to incompatible observables \cite{GlauberCoherent1963,SudarshanEquivalence1963}, and excludes descriptions of entanglement in terms of classical correlations \cite{SchrodingerEntanglement1935}.
Apart from being a mere foundational curiosity, such nonclassical features appear as key resources in quantum communication, quantum computation, and quantum metrology \cite{entanglement-review}.

In continuous variable (CV) systems, nonclassicality is clearly identified by negative values in quasiprobability distributions like the Wigner function, which contradicts the nonnegativity expected of a classical probability distribution \cite{quasiprobability-review}.
The presence of Wigner negativity in states or operations is a necessary resource for violations of Bell inequalities \cite{BanaszekNonlocality1998,BanaszekTesting1999,BraunsteinQuantum2005,WalschaersNonGaussian2021}, and for achieving quantum advantage in computational tasks \cite{wigner-computing-necessary}.
Specifically, the volume of the negative regions of the Wigner function quantifies nonclassicality \cite{AnatoleNegativity2004}, while its logarithm is a computable monotone in resource theories of non-Gaussianity and Wigner negativity \cite{resource-theory-neg-1,resource-theory-neg-2}.

Meanwhile, entanglement is present in states that cannot be generated with only classical communication and local state preparation, which results in correlated behaviours between distant parties that exceed classical correlations \cite{entanglement-review}.
Entanglement is similarly a necessary resource for Bell violations \cite{brunner_bell_2014}, quantum communication \cite{CurtyPRL04}, and quantum computation \cite{VidalEfficient2003,VandenNestUniversal2013}, while also a sufficient resource for certain quantum information processing tasks \cite{MasanesAll2006,masanes2008useful}.

Some relationships between entanglement and certain forms of nonclassicality are already known \cite{necessary-interferometric-entanglement,KimEntanglement2002,KilloranConverting2016,VogelUnified2014}, while recent works have begun to bridge Wigner negativity and entanglement \cite{DetectingZhang2013,GholipourShahandeh2016,WalschaersRemote2020,LiuExperimental,xiang2022distribution,JayachandranDynamics2023,ZawCertifiable2024,liu2024quantumentanglementphasespace}, clarifying the connections between different quantum effects.
However, such relationships between Wigner negativity and entanglement are known only for the bipartite case.
With more partitions, the entanglement structure becomes more complex, with genuine multipartite entanglement (GME) being the strongest form of entanglement.
GME excludes mixtures of states that are separable across any bipartition \cite{first-GME-paper}, and has been shown to be useful in distributed computing and quantum communication protocols \cite{GME-useful-distributed-computing,GME-useful-secret-sharing}.
This motivates GME criteria that extend beyond the bipartite setting.

\begin{figure*}[t]
    \includegraphics{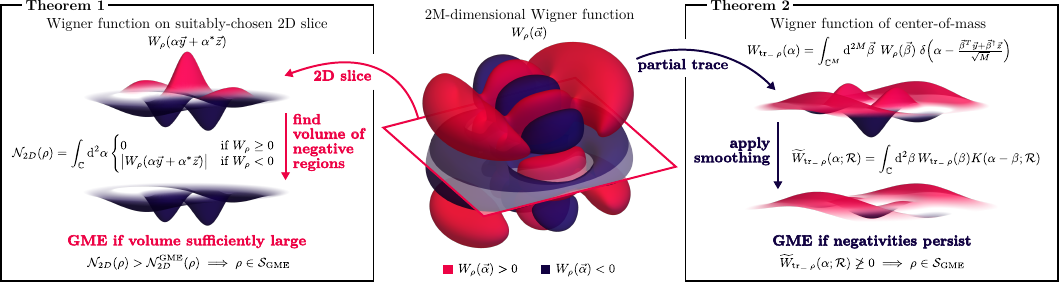}
    \caption{\label{fig:Overview}
        An illustration of our results that relate negativities in particular regions of the Wigner function to the presence of GME.
        The central figure is a three-dimensional slice of the Wigner function of an exemplary tripartite state specified in the End Matter.
        \cref{thm:GME-2D-slice} states that a sufficiently-large negativity volume along a suitably-chosen two-dimensional slice implies GME, while \cref{thm:GMN-reduced-state} states that the persistence of negativities in the center-of-mass Wigner function after a suitable smoothing process also implies GME.
    }
\end{figure*}

Such a need is particularly critical in hybrid qubit-CV architectures, which include circuit and cavity quantum-electrodynamics (cQED) \cite{cQED1,cQED2,cQED3,cQED4}, circuit quantum-acoustodynamics (cQAD) \cite{cQAD}, as well as trapped ions \cite{trappedIons1,trappedIons2} and atoms \cite{trappedAtoms}.
There, phase-space readouts, such as pointwise measurements of Wigner and characteristic functions, are not just commonplace, but in fact native to these setups \cite{cQED-parity-1,cQED-parity-2,cQED-CD-1,cQED-CD-2,cQED-CD-3,cQAD}.
In contrast, many existing CV GME criteria \cite{shalm2013three,ZhangGenuine2023,HyllusOptimal2006,TehPRA14,TehPRA22,Leskovjanova2025minimalcriteria,ToscanoSystematic2015} rely on quadrature measurements that are usually unavailable in such setups, thus necessitating state tomography or global transformations to recover quadrature data \cite{chou2025deterministic}.
Therefore, phase-space-based GME criteria tailored to these architectures, that can be implemented without quadrature measurements or full tomography, are still required in such systems.

In this work, we address this gap by showing that sufficient Wigner negativity, in two senses, implies GME in multimode CV systems.
The exact statements are illustrated in \cref{fig:Overview}, and will be specified in the following sections.
These findings are of notable foundational interest, as they demonstrate fundamental connections between two different nonclassical effects.

Next, we build upon these theorems to construct GME criteria that only require measuring either a finite two-dimensional region of the Wigner function, or a finite number of characteristic function points.
These have direct experimental implications, as they enable the detection of GME with only a few phase-space measurements.

\isection{Background and Definitions}%
An $M$-mode CV system is specified by annihilation operators $\vec{a} \coloneqq (a_1,a_2,\dots,a_M)$ that satisfy the canonical commutation relations $[a_m,a_{m'}] = 0$ and $[a_m,a_{m'}^\dagger] = \mathbbm{1} \delta_{m,m'}$.
Our central results will demonstrate that two different notions of nonclassicality in CV systems are fundamentally related.

The first notion of nonclassicality involves the Wigner function $W_{\rho}(\vec{\alpha}) = (2/\pi)^{M}{\tr}(\rho e^{i\pi\abs{\vec{a}-\vec{\alpha}}^2})$ of a state $\rho$, where $\vec{\alpha}=(\alpha_1, \dots, \alpha_M)$ is the vector of complex phase-space quadratures.
The Wigner function is a quasiprobability distribution in phase space in that it has all properties of a joint distribution of $\vec{\alpha}$---\emph{i.e.}, its marginal over the momentum is the position probability distribution of $\rho$, and vice versa over the position---except that it can take negative values \cite{quasiprobability-review}.
Hence, the presence of negativities in $W_\rho(\vec{\alpha})$ is one notion of nonclassicality, as it demonstrates that the observed behaviour cannot be simulated by a joint classical probability distribution of $\vec{\alpha}$.

The second notion of nonclassicality involves the concept of genuine multipartite entanglement (GME). A state $\rho$ is GME if it cannot be written as a convex combination of biseparable states, in the sense that
\begin{equation}
    \rho \;\text{ is GME } \implies \rho \neq \sum_{(\mathcal{A} \mid \bar{\mathcal{A}})} \sum_{k} p_{\mathcal{A}}^{(k)} \rho_{\mathcal{A}}^{(k)} \otimes \rho_{\bar{\mathcal{A}}}^{(k)} ,
\end{equation}
where $p_{\mathcal{A}}^{(k)} \geq 0$, $\sum_{(\mathcal{A} \mid \bar{\mathcal{A}})} \sum_{k} p_{\mathcal{A}}^{(k)}=1$ is a probability distribution,  $(\mathcal{A}|\bar{\mathcal{A}})$ runs over all bipartitions of the modes $\mathcal{A}=\{m_n\}_{n=1}^N$ and $\bar{\mathcal{A}}=\{m\}_{m=1}^M \setminus \mathcal{A}$ for $1\leq N<M$, and $\rho_{\{m_1, m_2, \ldots, m_N\}}^{(k)}$ are states defined locally on the $\{a_{m_1}, a_{m_2}, \ldots, a_{m_N}\}$ modes.
Such states cannot be prepared using only classical correlations and operations applied locally over bipartitions. Therefore, GME is another notion of nonclassicality, as it demonstrates the presence of correlations without a classical explanation.

{\it Primary Theoretical Results---}For our first theorem, we restrict ourselves to particular two-dimensional regions of phase space that faithfully capture correlations among the different modes.

\begin{theorem}[Sufficient Wigner negativity volume along a suitably-chosen two-dimensional slice implies GME]\label{thm:GME-2D-slice}
    Let $\circ$ be the element-wise product $[\mathbf{A}\circ\mathbf{B}]_{m,n} = [\mathbf{A}]_{m,n}[\mathbf{B}]_{m,n}$, and let $\vec{1} = (1,1,\dots,1)$ be a vector of ones.
    Choose some coefficients $\vec{y},\vec{z}\in\mathbb{C}^M$ such that $\vec{y}\circ\vec{y}^* - \vec{z}\circ\vec{z}^* = \vec{1}$, which specifies a two-dimensional slice $\{\alpha\vec{y} + \alpha^*\vec{z} : \alpha \in \mathbb{C}\}$ in phase space.
    Define the negativity volume of the Wigner function along this two-dimensional slice as
    \begin{equation}
        \mathcal{N}_{2D}(\rho)
        \coloneqq \pqty{\frac{\pi}{2}}^{M-1}\!\!
            \int_{\mathbb{C}}
            \dd[2]{\alpha}\left\{\begin{array}{lr}
                0 & \hspace{-6em}\text{if $W_{\rho}(\alpha\vec{y}+\alpha^*\vec{z})\geq0$,}\\[0.5ex]
                \abs\big{W_{\rho}(\alpha\vec{y}+\alpha^*\vec{z})} & \text{otherwise}.
            \end{array}\right.
    \end{equation}
    Then, $\mathcal{N}_{2D}(\rho) > \mathcal{N}_{2D}^{\operatorname{GME}}(\rho)$ implies that $\rho$ is GME, where
    \begin{equation}
        \mathcal{N}_{2D}^{\operatorname{GME}}(\rho)
        \coloneqq \frac{1}{4\sqrt{M-1}} -
        \frac{\pi^{M-1}}{2^M} \!\! \int_{\mathbb{C}}
        \dd[2]{\alpha} W_{\rho}(\alpha \vec{y}+\alpha^* \vec{z}).
    \end{equation}
\end{theorem}
The proof is given in Sec.~S1 of the Supplemental Material \cite{supplement}, \nocite{Bochner-theorem-1,Bochner-theorem-2}
where it is also shown that
$0 \leq \mathcal{N}_{2D}^{\operatorname{GME}}(\rho) \leq \max_{\rho,\vec{y},\vec{z}}\mathcal{N}_{2D}^{\operatorname{GME}}(\rho) = (4\sqrt{M-1})^{-1} + (2M)^{-1}$, so the GME bound is nonnegative and finite for all states.
Therefore, a sufficiently-large Wigner negativity volume---\emph{i.e.}, the volume occupied by the negative regions of $W_\rho(\vec{\alpha})$---along the two-dimensional slice $\{\alpha\vec{y}+\alpha^*\vec{z}:\alpha\in\mathbb{C}\}$ of phase space implies GME.
In the case that $\mathcal{N}_{2D}^{\operatorname{GME}}(\rho)$ is not known or difficult to compute, the uniform bound $(4\sqrt{M-1})^{-1} + (2M)^{-1}$ can be used instead.
Furthermore, \cref{thm:GME-2D-slice} reduces to a bipartite entanglement witness from Ref.~\cite{liu2024quantumentanglementphasespace} when $M=2$.

Our next theorem concerns the center-of-mass mode
\begin{equation}
    a_+ \coloneqq \frac{1}{\sqrt{M}}\sum_{m=1}^{M}\pqty{y_m a_m + z_m a_m^\dag},
\end{equation}
where $\vec{y},\vec{z}\in\mathbb{C}^M:\vec{y}\circ\vec{y}^*-\vec{z}\circ\vec{z}^* = \vec{1}$ as before.
Given a state $\rho$ describing the full $M$-mode system, the reduced state $\tr_{-}\rho$ describing its center-of-mass is the partial trace over the relative modes $\{a_{-m}\}_{m=2}^M$ such that
\begin{equation}
    \tr_{-}\rho \coloneqq \tr_{a_{-2}}\tr_{a_{-3}}\cdots\tr_{a_{-M}}\rho,
\end{equation}
where the collection of modes $\{a_{+}\}\cup\{a_{-m}\}_{m=2}^M$ satisfy the canonical commutation relations.

Here, $a_{+}$ is related via local transformations to $\propto \sum_{m=1}^M a_m$, the eponymous mode that describes the center-of-mass of $M$ identically-coupled trapped ions \cite{trapped-ion-center-of-mass}.
Hence, $\tr_{-}\rho$ is a description of the system that ignores every degree of freedom of the system except for its center-of-mass motion.
The Wigner function of $\tr_-\rho$ can also be computed from the full Wigner function by marginalizing over the relative degrees of freedom as
\begin{equation}
    W_{\tr_-\rho}(\alpha) = \int_{\mathbb{C}^{M}}\dd[2M]{\vec{\beta}}
        \;W_{\rho}(\vec{\beta})
        \;\delta\pqty{
            \alpha -
            \tfrac{
                \vec{\beta}^T\vec{y} +
                \vec{\beta}^{\dagger}\vec{z}
            }{\sqrt{M}}
        }.
\end{equation}
With this in mind, we can state the following theorem.

\begin{theorem}[Negativity of the smoothed Wigner function of the center-of-mass implies GME]\label{thm:GMN-reduced-state}
    Choose $M-2$ states $\mathcal{R} = \{\varrho_m\}_{m=1}^{M-2}$. Define the smoothed Wigner function of the center-of-mass of the system as
\begin{equation}
\widetilde{W}_{\tr_{-}\rho}(\alpha;\mathcal{R})\\
    \coloneqq \int_{\mathbb{C}}\dd[2]{\beta}
    \;W_{\tr_{-}\rho}(\beta)
    \;K\pqty{\alpha-\beta;\mathcal{R}},
\end{equation}
where $K(\alpha,\mathcal{R})$ is the convolution kernel
\begin{equation}
\begin{aligned}
    K(\alpha;\mathcal{R})
    &\coloneqq
    \int_{\mathbb{C}^{M-2}}\dd[2(M-2)]{\vec{\gamma}}
    \;\prod_{m=1}^{M-2}W_{\varrho_m}(\gamma_m)\\[-1ex]
    &\hspace{4.5em}{}\times{}2\pqty{1-M^{-1}}
    \;\delta\pqty{
        \alpha-\tfrac{\vec{\gamma}^T\vec{1}}{\sqrt{M}}
    }.
\end{aligned}
\end{equation}
Then, the smoothed Wigner function lower bounds, up to a factor, the trace distance to all non-GME states as
\begin{equation}
\begin{aligned}
    \max\Bqty{0,-\widetilde{W}_{\tr_{-}\rho}(\alpha;\mathcal{R})} \leq \frac{2}{\pi}\min_{\sigma \notin \operatorname{GME}} \| \sigma - \rho \|_1.
\end{aligned}
\end{equation}
Hence, $\exists\alpha : \widetilde{W}_{\tr_{-}\rho}(\alpha;\mathcal{R}) < 0$ implies that $\rho$ is GME.
\end{theorem}
Here, $\|\bullet\|_1$ is the trace norm, and the proof is laid out in full in Sec.~S2 of the Supplemental Material \cite{supplement}.
Therefore, if the negativities in the center-of-mass Wigner function persist even after smoothing it with an appropriate filter function, there must be GME.
As before, \cref{thm:GMN-reduced-state} also reduces to another bipartite witness studied in Refs.~\cite{GholipourShahandeh2016,JayachandranDynamics2023,ZawCertifiable2024} for $M=2$.

A simple choice for the filter is to take $\mathcal{R}=\mathcal{R}_G$ to be Gaussian states.
Then, up to translations $\alpha \to \alpha+\alpha'$,
\begin{equation}\label{eq:filter-function-Gaussian}
K(\alpha;\mathcal{R}_G) = \tfrac{1-M^{-1}}{\pi\sqrt{\det\Sigma'}} e^{-\frac{1}{2}\abs{{\Sigma'}^{-\frac{1}{2}}{\scriptscriptstyle\spmqty{\Re[\alpha]\\\Im[\alpha]}}}^2} : \sqrt{\det\Sigma'} \geq \tfrac{M-2}{4M}.
\end{equation}
This also allows us to recast the theorem by relating GME to another preexisting notion of nonclassicality in the literature due to \citet{nonclassicality-depth}.
By substituting \cref{eq:filter-function-Gaussian} into \cref{thm:GMN-reduced-state}, with detailed steps in Sec.~S3 of the Supplemental Material \cite{supplement}, we obtain the following:

\begin{corollary}[Sufficient nonclassicality depth of the center-of-mass implies GME]\label{col:nonclassicality-depth-reduced-state}
    The nonclassicality depth $\tau_c$ of a state $\rho$ is defined as \cite{nonclassicality-depth}
    \begin{equation}
        \tau_c(\rho) \coloneqq \min\!\Bqty{\tau : \forall \alpha : \frac{1}{\pi\tau}\int_{\mathbb{C}}\dd[2]{\beta}
        \;P_\rho(\beta)
        \;e^{-\frac{\abs{\alpha-\beta}^2}{\tau}} \geq 0 },
    \end{equation}
    where $P_\rho(\alpha)$ is the Glauber $P$ function of $\rho$ such that $\rho = \int_{\mathbb{C}}\dd[2]{\alpha} P_{\rho}(\alpha)\ketbra{\alpha}$, and $\ket{\alpha}$ is the coherent state.
    Then, $\tau_c(\tr_{-}\rho) > 1 - M^{-1}$
    implies that $\rho$ is GME.
\end{corollary}

Curiously, the converse of \cref{col:nonclassicality-depth-reduced-state} also provides a sufficient condition for generating GME by interfering a state with the vacuum via a maximally-mixing interferometer $U$.
Here, $U$ transforms $\vec{a}$ as $U^\dag \vec{a} U = \mathbf{U}\vec{a}$, where $\mathbf{U}$ is an $M \times M$ unitary matrix such that $\forall m : \abs{[\mathbf{U}]_{1,m}}^2=M^{-1}$, and can be constructed out of two-mode beamsplitters \cite{ReckExperimental1994}.
Then, interfering an adequately nonclassical state with the vacuum via $U$ is \emph{sufficient} for GME:
\begin{equation}
    \tau_c(\rho_1) > 1-M^{-1} \!\implies\! U\pqty{\rho_1 \otimes \ketbra{0}^{\otimes(M-1)}}U^\dag \in \operatorname{GME}.
\end{equation}
This complements the long-known result that a nonzero nonclassicality depth is \emph{necessary} for generating entanglement via interference with the vacuum, \emph{i.e.}, $\tau_c(\rho_1) = 0 \implies U(\rho_1 \otimes \ketbra{0}^{\otimes(M-1)})U^\dag \notin \operatorname{GME}$ \cite{necessary-interferometric-entanglement}, and also provides a more readily computable condition for arbitrary mixed states than the condition given in Ref.~\cite{VogelUnified2014}.

\isection{Construction of GME criteria}%
Direct pointwise measurements of the Wigner function are routinely implemented in qubit-CV systems described by the Jaynes--Cummings interaction.
Building upon \cref{thm:GME-2D-slice}, we can construct a GME criterion which relies only on such Wigner function measurements performed over a finite region of phase space.

\begin{figure}
    \centering
    \includegraphics{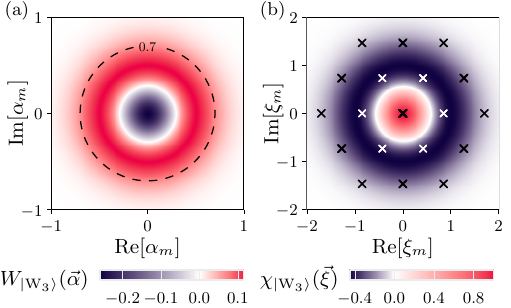}
    \caption{
        \label{fig:W-state-combined}(a) The Wigner function $W_{\ketW[3]}(\vec{\alpha})$ of the tripartite W state along the slice $\alpha_1=\alpha_2=\alpha_3$.
            The GME of this state can be certified using \cref{col:GME-Wigner-witness} by just integrating $W_{\ketW[3]}(\vec{\alpha})$ along this two-dimensional slice over the region $0 \leq \abs{\alpha_m} \lesssim r$ for any $r \gtrsim 0.7$.
        (b) The characteristic function $\chi_{\ketW[3]}(\vec{\xi})$ of $\ketW[3]$ along $\xi_1=\xi_2=\xi_3$.
        Its GME is certified using \cref{col:GME-characteristic-witness} by just measuring $\chi_{\ketW[3]}(\vec{\xi})$ at 10 of the 19 points marked out as crosses, with the other 9 values obtained from the symmetry $\chi_{\ketW[3]}(-\vec{\xi}) = \chi_{\ketW[3]}^*(\vec{\xi})$.
        Details of the plotted state are given in the End Matter.
    }
\end{figure}

\begin{corollary}[GME criterion with Wigner function measurements over a finite region]\label{col:GME-Wigner-witness}
    Let the absolute volume of the Wigner function on $\{\alpha\vec{y} + \alpha^*\vec{z} : \alpha \in \omega\}$, where $\vec{y}\circ\vec{y}^*-\vec{z}\circ\vec{z}^* = \vec{1}$ and $\omega \subseteq \mathbb{C}$ is Lebesgue-measurable, be
    \begin{equation}\label{eq:GMEWignerFiniteRegionDef}
    \begin{aligned}
        \mathcal{V}_{2D}(\rho;\omega) &\coloneqq \pqty{\frac{\pi}{2}}^{M-1}\int_{\omega}\dd[2]{\alpha} \abs\big{W_{\rho}\pqty{\alpha\vec{y} + \alpha^*\vec{z}}}.
    \end{aligned}
    \end{equation}
    Then, $\mathcal{V}_{2D}(\rho;\omega) > (2\sqrt{M-1})^{-1}$ implies that $\rho$ is GME.
\end{corollary}
The proof is given in Sec.~S4 of the Supplemental Material \cite{supplement}.
As an example, consider the tripartite W state $\ketW[3] \propto \ket{100} + \ket{010} + \ket{001}$, whose Wigner function along the two-dimensional slice $\alpha\vec{1} = (\alpha,\alpha,\alpha)$ is shown in \cref{fig:W-state-combined}(a).
The absolute volume of the plotted slice over $\omega_r \coloneqq \{\alpha : \abs{\alpha} \leq r\}$ given a radius $r$ is
\begin{equation}
    \mathcal{V}_{2D}(\ketW[3];\omega_{0.7}) \gtrsim \pqty{2\sqrt{2}}^{-1},
\end{equation}
where the right-hand side is the GME bound from \cref{col:GME-Wigner-witness} with $M=3$.
Therefore, integrating over this finite region with any $r > 0.7$ certifies the GME of $\ketW[3]$.

In practice, \cref{eq:GMEWignerFiniteRegionDef} would be computed with numerical integration using only measurements of the Wigner function at a finite number of phase space points in $\omega$.
Some practical issues and error analysis of this pragmatic approach are discussed in our companion paper \cite{CV-GME-paper}.
We found that to certify the GME of the W state with discrete points of phase space, a rigorous approach demands Wigner function measurements at $\approx 76\,000$ points, while a heuristic approach commonplace in the numerical integration literature requires $\approx 78$ points.

Complementing the above criterion, consider pointwise measurements of the characteristic function $\chi_{\rho}(\vec{\xi})$, which is the Fourier transform of the Wigner function as
\begin{equation}
    \chi_{\rho}(\vec{\xi}) \coloneqq \int\dd[2M]{\vec{\alpha}}
    \;W_\rho(\vec{\alpha})
    \;e^{\sum_{m=1}^M(\alpha_m^*\xi_m-\alpha_m\xi_m^*)}.
\end{equation}
Characteristic function measurements are routinely performed in qubit-CV systems, especially when the qubit and CV modes are weakly and dispersively coupled.

Using analogous techniques to the bipartite case \cite{ZawCertifiable2024} detailed in Sec.~S5 of the Supplemental Material \cite{supplement}, we can construct a GME criterion which relies only on a finite number of characteristic function measurements.
\begin{corollary}[GME criterion with characteristic function measurements at finitely many points]\label{col:GME-characteristic-witness}
    Choose $N$ phase-space points $\Xi = \{\xi_n\}_{n=1}^N$ and coefficients $\vec{y},\vec{z}\in\mathbb{C}^M : \vec{y}\circ\vec{y}^* - \vec{z}\circ\vec{z}^* = \vec{1}$.
    Construct the matrix $\mathbf{C}(\rho;\Xi) \in \mathbb{C}^{N \times N}$ as
    \begin{equation}
    \begin{aligned}
        [\mathbf{C}(\rho;\Xi)]_{n,n'}
        &\coloneqq \frac{1}{N}\chi_{\rho}\pqty\big{
            (\xi_n-\xi_{n'})\vec{y} +
            (\xi_n^*-\xi_{n'}^*)\vec{z}
        } \\
        &= \frac{1}{N}\chi_{\tr_{-}\rho}\pqty{
            \sqrt{M}(\xi_n-\xi_{n'})
        }.
    \end{aligned}
    \end{equation}
    Next, choose $M-2$ states $\mathcal{R} = \{\varrho_m\}_{m=1}^{M-2}$ and construct $\mathbf{K}(\mathcal{R};\Xi)\in\mathbb{C}^{N\times N}$ by computing
    $[\mathbf{K}(\mathcal{R};\Xi)]_{n,n'} \coloneqq \prod_{m=1}^{M-2} \chi_{\varrho_m}(\xi_n-\xi_{n'})$.
    Then, the magnitude of the most negative eigenvalue of their element-wise product
    \begin{equation}
        \mathcal{N}_C(\rho;\Xi,\mathcal{R}) \coloneqq \max\Bqty{0,-\operatorname{mineig}\bqty{
        \mathbf{C}(\rho;\Xi)\circ\mathbf{K}(\mathcal{R};\Xi)
        }}
    \end{equation}
    lower bounds the trace distance to all non-GME states as
    \begin{equation}
        \mathcal{N}_C(\rho;\Xi,\mathcal{R}) \leq \min_{\sigma \notin\mathrm{GME}} \| \sigma - \rho \|_1.
    \end{equation}
    Hence, $\mathbf{C}(\rho;\Xi)\circ\mathbf{K}(\mathcal{R};\Xi) \not\succeq 0$ implies that $\rho$ is GME.
\end{corollary}

This GME criterion requires measurements of less than $N^2$ points of the characteristic function along the two-dimensional slice $\{\xi\vec{y} + \xi^*\vec{z} : \xi \in \mathbb{C}\}$ of phase space.
Furthermore, it can be implemented with either local  $\chi_\rho(\vec{\xi})$ or center-of-mass $\chi_{\tr_{-}\rho}(\xi)$ characteristic function measurements.
The latter can be directly implemented in trapped ion systems by coupling the readout laser to a normal mode of the collective motion of the ions \cite{trappedIons1,trappedIons2}.

As an example, take the tripartite W state from before. Its characteristic function along the two-dimensional slice $\xi\vec{1}$ is plotted in \cref{fig:W-state-combined}.
In order to obtain $\mathcal{N}_C\pqty{\ketW[3];\Xi,\mathcal{R}}$ for $\Xi = \{0,\pm (\xi_0 + \xi_0^*),\pm \xi_0,\pm \xi_0^*\}$ with $\xi_0 = (85 + i147)/200$, values of the characteristic function at the points crossed out in \cref{fig:W-state-combined} must be found.
Due to the symmetry $\chi_\rho(-\vec{\xi}) = \chi_\rho^*(\vec{\xi})$, just ten points $\{\xi_n-\xi_{n'}\}_{n \geq n'}$ in phase space have to be actually measured to obtain $\mathcal{N}_C\pqty{\ketW[3];\Xi,\{\ket{0}\}} = 0.0176 > 0$, which certifies the GME of the tripartite W state using \cref{col:GME-characteristic-witness}.

Further detailed examples, case studies of detected states, and the impact of practical effects like losses can be found in our companion paper \cite{CV-GME-paper}.
Families of states detected by both \cref{thm:GME-2D-slice,thm:GMN-reduced-state} are given for the full range of $3 \leq M < \infty$.
However, we also found that the robustness decreases with the number of modes.
This shows that our criteria can detect GME for any finite $M$, although they will be challenging to implement for very large $M$ in the presence of noise.

\isection{Conclusion}%
In this work, we established two theorems that relate the Wigner negativity of a multimode continuous-variable system with the presence of genuine multipartite entanglement (GME).
The first theorem states that a sufficiently-large negativity volume along a particular two-dimensional slice of the Wigner function implies GME, while the second theorem states that the presence of Wigner negativity of the center-of-mass mode, even after smoothed with a suitably-chosen filter function, implies GME.
Quantitatively, the pointwise value of the smoothed Wigner function in the latter theorem also bounds the geometric distance between the detected state and the set of non-GME states.

By themselves, these theorems are notable as they fundamentally link two different notions of nonclassicality: one in the sense of quasiprobabilities, and the other in the sense of correlations.
A consequence of the second theorem also identifies sufficient conditions to generate GME by interfering a state with the vacuum via maximally-mixing multimode interferometers, which complement necessary conditions well-known in the literature.

Beyond the above foundational interests, our findings also have important implications in the field of circuit and cavity quantum electrodynamics, circuit quantum acoustodynamics, and trapped ions and atoms.
In such systems, direct measurements of Wigner or characteristic functions are routinely performed, and sometimes constitute the native readout available in the specific experimental platform.
Using our theorems, we show that it is possible to construct GME criteria that rely only on performing Wigner function measurements over a finite region, or only characteristic function measurements of a finite number of points, both over a suitably-chosen two-dimensional slice of phase space.
Our criteria are therefore easily implementable in such systems, where existing continuous-variable GME criteria can be difficult to implement due to the unavailability of direct quadrature measurements.
In our companion paper, we further extend these findings to construct more GME criteria that rely on controlled-unitary operations and qubit measurements available in these architectures \cite{CV-GME-paper}.

Our theorems also raise many interesting questions that could be possible avenues for further research.
Foundationally, we anticipate that similar techniques could be developed to relate nonclassicality to other types or structures of entanglement, like entanglement depth or non-passive entanglement.
Operationally, while we provide some rules of thumb in the companion paper on the choice of a 2D slice to detect GME using our criteria, the optimal choice for a general state remains an open question.
It might also be possible to extend our criteria to detect even more GME states, for example by looking at multiple 2D slices at once, or even considering higher dimensional slices.
However, these considerations are beyond the scope of our current work, and we leave it to future efforts in addressing these questions.

\isection{Acknowledgments}%
This work is supported by the National Research Foundation, Singapore, under its Centre for Quantum Technologies Funding Initiative (S24Q2d0009), the National Natural Science Foundation of China (Grants No. 12125402, No. 12534016, No. 12405005, No. 12505010), the Beijing Natural Science Foundation (Grant No. Z240007), and the Quantum Science and Technology-National Science and Technology Major Project (No. 2024ZD0302401, No. 2021ZD0301500).
J.G. acknowledges the support of the Postdoctoral Fellowship Program of CPSF (No. GZB20240027), and the China Postdoctoral Science Foundation (No. 2024M760072).
M.F. was supported by the Swiss National Science Foundation Ambizione Grant No. 208886, and by The Branco Weiss Fellowship -- Society in Science, administered by the ETH Z\"{u}rich.
S.L. acknowledges the China Postdoctoral Science Foundation (No. 2023M740119).

\bibliography{refs}

\onecolumngrid
\begin{center}\textbf{End Matter}\end{center}
\twocolumngrid

\isection{Details of plotted state in \cref{fig:Overview}}%
The state chosen to illustrate our theorems is given by
\begin{equation}
\begin{aligned}
|\psi\rangle= & \frac{1}{5 \sqrt{2}}\left(a_{+}^{\dagger}+\frac{a_{+}^{\dagger 3}}{\sqrt{3!}}\right)\left(1+\sqrt{19} a_{-}^{\dagger}\right)|000\rangle \\
& +\frac{1}{\sqrt{10}}\left(\frac{a_{+}^{\dagger2}}{\sqrt{2!}}+\frac{a_{+}^{\dagger4}}{\sqrt{4!}}\right)\left(\frac{a_{-}^{\dagger2}}{\sqrt{2!}}\right)|000\rangle
\end{aligned}
\end{equation}
where we defined $\sqrt{3}a_{+} \coloneqq \sum_{m=1}^3 a_m$ and $\sqrt{6}a_{-} \coloneqq 2 a_1 - a_2 - a_3$.
Analytical expressions of Wigner functions of states given in the Fock basis are known \cite{WignerMatrixElements}, from which $W_{\ket{\psi}}(\vec{\alpha})$ can be directly computed.
To simplify the notation for the Wigner function coordinates, we similarly introduce $\sqrt{3}\alpha_+ \coloneqq \sum_{m=1}^3\alpha_m$ and $\sqrt{6}\alpha_- \coloneqq 2\alpha_1-\alpha_2-\alpha_3$.
The resulting Wigner function is plotted for the three-dimensional cut $\Im[2\alpha_1-\alpha_2-\alpha_3] = 0$ and $\alpha_2=\alpha_3$ in \cref{fig:Overview} with the axes $\Re[\alpha_+]$ (\emph{left-to-right of page}), $\Im[\alpha_+]$ (\emph{towards the page}), and $\Re[\alpha_-]$ (\emph{bottom-to-top of page}).

From the two-dimensional slice $\vec{\alpha} = \alpha\vec{1}$, which corresponds to setting $\Re[\alpha_-] = 0$ in the figure, $\mathcal{N}_{2D}^{\text{GME}}(\ket{\psi}) = (75\sqrt{2}+56)/600 < 0.2702 < \mathcal{N}_{2D}(\ket{\psi})$.
Therefore, $\ket{\psi}$ is detected by \cref{thm:GME-2D-slice}.

Meanwhile, when smoothing the Wigner function of its center-of-mass mode with the kernel $K(\alpha) = 8e^{-6\abs{\alpha}^2}/\pi$ from \cref{eq:filter-function-Gaussian}, the resulting function yields the negative value $\widetilde{W}_{\ket{\psi}}(0) = -7/16\pi$ at the origin. Therefore, the GME of $\ket{\psi}$ can also be detected by \cref{thm:GMN-reduced-state}.

\isection{Details of plotted state in \cref{fig:W-state-combined}}%
The Wigner function of the tripartite W state, again found using the analytical expression of Wigner functions given in the Fock basis \cite{WignerMatrixElements}, is $W_{\ketW[3]}(\vec{\alpha}) = (2/\pi)^3  (4|\vec{\alpha}\cdot\vec{1}|^2/3 - 1)e^{-2\abs{\alpha}^2}$. Its absolute volume can be directly calculated to be
\begin{equation}
    \mathcal{V}_{2D}(\ketW[3];\omega_{r}) = \begin{cases}
        \frac{(1+12r^2)e^{-6r^2}-1}{3} & \text{if $r < \frac{1}{2\sqrt{3}}$,}\\
        \frac{4e^{-1/2}-(1+12r^2)e^{-6r^2}-1}{3} & \text{otherwise.}
    \end{cases}
\end{equation}
Meanwhile, its characteristic function is $\chi_{\ketW[3]}(\vec{\xi}) = (1-|\vec{\xi}\cdot\vec{1}|^2/3)e^{-\abs{\vec{\xi}}^2/2}$.
For the chosen $\Xi = \{0,\pm(\xi_0+\xi_0^*),\pm\xi_0,\pm\xi_0^*\}$, the 19 plotted points $\{\xi-\xi' : \xi,\xi' \in \Xi\}$ are
\begin{equation}
\begin{aligned}
    \big\{
        0,
        \pm \xi_0,
        \pm \xi_0^*,
    &\pm 2\xi_0,
        \pm 2\xi_0^*,
        \pm \xi_0 \pm \xi_0^*,
        \pm \xi_0 \mp \xi_0^*,\\
    &\pm 2\xi_0 \pm \xi_0^*,
        \pm \xi_0 \pm 2\xi_0^*,
        \pm 2\xi_0 \pm 2\xi_0^*
    \big\}
\end{aligned}
\end{equation}

\onecolumngrid

\clearpage

\phantomsection\addcontentsline{toc}{part}{Supplemental Material for: Sufficient Wigner Negativity Implies Genuine Multipartite Entanglement}
\title{Supplemental Material for: Sufficient Wigner Negativity Implies Genuine Multipartite Entanglement}

\maketitle
\onecolumngrid

\setcounter{section}{0}
\setcounter{figure}{0}
\setcounter{equation}{0}

\renewcommand{\thesection}{S\arabic{section}}
\renewcommand{\thefigure}{S\arabic{figure}}
\renewcommand{\theequation}{S\arabic{equation}}

\section{\label{apd:GME-2D-slice}Proof of Theorem~\ref*{thm:GME-2D-slice}}
We begin by proving the following lemma
\begin{lemma}\label{lem:Wigner-absolute-volume-all}
For all $M$-mode states $\rho$ and $\vec{y},\vec{z}\in\mathbb{C}^M$ such that $\vec{y}\circ\vec{y}^*-\vec{z}\circ\vec{z}^* = \vec{1}$,
\begin{equation}\label{eq:Criterion2lemma}
\begin{aligned}
    \abs{
        \int_{\mathbb{C}}\dd[2]{\alpha} W_\rho(\alpha\vec{y}+\alpha^*\vec{z})
    } &\leq \frac{1}{M}\pqty{\frac{2}{\pi}}^{M-1}, &
\int_{\mathbb{C}}\dd[2]{\alpha} \abs{W_\rho(\alpha\vec{y}+\alpha^*\vec{z})}^2 &\leq \frac{1}{2M}\pqty{\frac{2}{\pi}}^{2M-1}.
\end{aligned}
\end{equation}
\end{lemma}
\begin{proof}
    Let us first prove this for $\vec{y} = \vec{1}$ and $\vec{z} = (0,0,\dots,0)$.
    Consider the collective modes $\vec{a}_{+} = (a_{+1},a_{+2},\dots,a_{+M}) \coloneqq \mathbf{U}\vec{a}$, where $\mathbf{U}$ is an $M \times M$ unitary matrix chosen such that $\sqrt{M}a_{+} \coloneqq \sqrt{M}a_{+1} = \vec{1}^\dagger\vec{a} = \sum_{m=1}^{M}a_m$.
    Then, the Wigner function $W^{(+)}_{\rho}(\vec{\alpha}_+)$ of $\rho$ written in terms of the collective coordinates $\vec{\alpha}_+ = (\alpha_{+m})_{m=1}^M$ is
    \begin{equation}
        W^{(+)}_{\rho}(\vec{\alpha}_+) = \pqty{\frac{2}{\pi}}^{M} \tr(
        \rho \; e^{i\pi \abs{\vec{a}_+-\vec{\alpha}_+}^2}
        ).
    \end{equation}
    The key first step is to notice that
    \begin{equation}\label{eq:Wigner-same-arg-is-collective}
    \begin{aligned}
        W_\rho(\alpha\vec{1})
        &= \pqty{\frac{2}{\pi}}^M\tr(\rho \;
            e^{i\pi\abs{
                \vec{a} - \alpha\vec{1}
            }^2}
        ) = \pqty{\frac{2}{\pi}}^M\tr(\rho \;
            e^{i\pi\pqty{
                    \abs{\vec{a}}^2 -
                    \alpha^*\vec{1}^\dagger\vec{a} -
                    \alpha \vec{a}^\dag\vec{1} +
                    M\abs{\alpha}^2
            }}
        ) \\
        &= \pqty{\frac{2}{\pi}}^M\tr(\rho \;
            e^{i\pi\pqty{
                    \abs{
                        \vec{1}^\dag\vec{a}/\sqrt{M} -
                        \sqrt{M}\alpha
                    }^2 +
                    \abs{\vec{a}}^2 -
                    \abs{\vec{1}^\dag\vec{a}/\sqrt{M}}^2
            }}
        )\\
        &= \pqty{\frac{2}{\pi}}^M\tr(\rho \; e^{i\pi\pqty{
            \abs{a_+ - \sqrt{M}\alpha}^2 +
            \abs{\vec{a}_+}^2 -
            \abs{\alpha_{+}}^2
        }}) \\
        &= \pqty{\frac{2}{\pi}}^M\tr(\rho \; e^{i\pi\pqty{
            \abs{a_+ - \sqrt{M}\alpha}^2 +
            \sum_{m=2}^{M}\abs{\alpha_{+m}}^2
        }}) \\
        &= W_\rho^{(+)}(\sqrt{M}\alpha,0,\cdots,0).
    \end{aligned}
    \end{equation}
    Now, let us split the operators in the expectation value into the displaced parity operator $\Pi_+(\alpha) \coloneqq {\exp}(i\pi |a_{+}-\alpha|^2)$ of the center-of-mass mode and the parity operator $\Pi_- \coloneqq {\exp}(i\pi\sum_{m=2}^{M} \abs{a_{+m}}^2)$ of the other collective modes.
    We shall also denote the partial trace over the center-of-mass mode (relative modes) as $\tr_+$ ($\tr_-$).
    Then, \cite{quasiprobability-review}
    \begin{equation}
    \begin{aligned}
        W_\rho(\alpha\vec{1}) &= W_\rho^{(+)}(\sqrt{M}\alpha,0,\dots,0) \\
        &= \pqty{\frac{2}{\pi}}^{M} \tr[
            \Pi_+(\sqrt{M}\alpha) \Pi_- \rho
        ] \\
        &= \pqty{\frac{2}{\pi}}^{M} \tr_+\!\bqty\Big{
            \Pi_+(\sqrt{M}\alpha) \underbrace{\tr_-(\Pi_- \rho)}_{=:R_+}
        } \\
        &= \frac{1}{2}\pqty{\frac{2}{\pi}}^{M} W^{(+)}_{R_+}(\sqrt{M}\alpha),\\
    \end{aligned}
    \end{equation}
    where $R_+$ is a Hermitian operator defined on the $a_{+}$ mode.
    In other words, we can treat the $M$-mode Wigner function $W_\rho(\alpha\vec{1})$ of $\rho$ on the local modes as a single-mode Wigner function $W^{(+)}_{R_+}(\sqrt{M}\alpha)$ of $R_+$ on the center-of-mass mode.
    Note here that we are treating $R_+$ as a generic observable, instead of a state, since $R_+$ is not positive semidefinite in general.
    This accounts for the difference in a factor of $\pi$ in the normalization of $W_{R_+}^{(+)}(\alpha)$ compared to the normalization of the Wigner function of a state \cite{quasiprobability-review}.
    This means that
    \begin{equation}
    \begin{aligned}
        \abs{
            \int_{\mathbb{C}}\dd[2]{\alpha} W_\rho(\alpha\vec{1})
        }
        &= \abs{
            \frac{1}{2}\pqty{\frac{2}{\pi}}^M \int_{\mathbb{C}}\dd[2]{\alpha} W^{(+)}_{R_+}(\sqrt{M}\alpha)
        }
        = \frac{1}{M}\pqty{\frac{2}{\pi}}^{M-1}
        \abs{
            \frac{1}{\pi}\int_{\mathbb{C}}\dd[2]{\alpha} W^{(+)}_{R_+}(\alpha)
        }\\
        &= \frac{1}{M}\pqty{\frac{2}{\pi}}^{M-1}
        \abs\big{
            \tr(R_+)
        }
        = \frac{1}{M}\pqty{\frac{2}{\pi}}^{M-1}
        \abs\big{
            \tr(\rho \Pi_-)
        }
        \leq \frac{1}{M}\pqty{\frac{2}{\pi}}^{M-1},
    \end{aligned}
    \end{equation}
    where we have used the trace of an operator expressed in terms of its Wigner function \cite{quasiprobability-review}, and the last inequality comes from the fact that $\Pi_-$ is a unitary operator whose eigenvalues have modulus one.
    From the last inequality, it is also clear that this bound is tight for any state $\ket{\psi}_+\otimes\ket{\phi_{\pm}}_-$ where $\ket{\phi_{\pm}}_-$ is an eigenstate of $\Pi_-$ with eigenvalue $\pm 1$, since for which $\int_{\mathbb{C}}\dd[2]{\alpha} W_{\ket{\psi}_+\otimes\ket{\phi_{\pm}}_-}(\alpha) = \pm M^{-1}(2/\pi)^{M-1}$.

    Next, we have
    \begin{equation}\label{eq:Rplus-inner-product}
        \int_{\mathbb{C}}\dd[2]{\alpha} \abs{W_\rho^{(+)}(\alpha\vec{1})}^2
        = \frac{1}{4M}\pqty{\frac{2}{\pi}}^{2M} \int_{\mathbb{C}}\dd[2]{\alpha} W_{R_+}^{(+)}(\alpha) W_{R_+}^{(+)}(\alpha)
        = \frac{1}{2M}\pqty{\frac{2}{\pi}}^{2M-1} \tr(R_+^2),
    \end{equation}
    where we have used the correspondence between inner products in the Wigner function and Hilbert spaces \cite{quasiprobability-review}.
    Now, identifying $\tr(R_+^2) = \|R_+\|_2^2$ as the square of the Hilbert-Schmidt norm, \cref{eq:Rplus-inner-product} can be related to the trace norm by the inequality $\|R_+\|_2^2 \leq \|R_+\|_1^2 = (\tr|R_+|)^2$. This can be further bounded using the variational definition of the trace norm, and noting that $-\mathbbm{1}_- \preceq \Pi_- \preceq \mathbbm{1}_-$, as
    \begin{equation}
    \begin{aligned}
        \tr|R_+| &&= \sup_{-\mathbbm{1}_+\preceq M_+ \preceq\mathbbm{1}_+} \tr_+[M_+ R_+] &
        &= \sup_{-\mathbbm{1}_+\preceq M_+ \preceq\mathbbm{1}_+} \tr[M_+\Pi_-\rho] &
        &\leq \sup_{-\mathbbm{1}\preceq M \preceq\mathbbm{1}} \tr[M \rho] &
        &= \tr|\rho| = 1.
    \end{aligned}
    \end{equation}
    Notice that this bound is tight for the state $\ket{\psi}_+\otimes\ket{\phi_{\pm}}_-$ from before, since $R_+ = \ketbra{\psi}_+ \tr_-(\Pi_-\ketbra{\phi_{\pm}}_-) = \pm \ketbra{\psi}_+ \implies \tr(R_+^2) = 1$.
    In particular, this is true for $\ket{\psi}_+ = \ket{0}$ and $\ket{\phi}_- = \ket{0}^{\otimes M-1}$.

    Finally, to complete the proof, we need to show that the above statements hold true with the replacement $W_\rho(\alpha\vec{1}) \to W_\rho(\alpha\vec{y}+\alpha^*\vec{z})$ for all $\vec{y},\vec{z} \in \mathbb{C}^M$ that satisfy $\vec{y}\circ\vec{y}^* - \vec{z}\circ\vec{z}^* = \vec{1}$.
    From the restriction $\forall m: |y_m|^2 - |z_m|^2 = 1$, the quantities $y_m$ and $z_m$ can be parameterized as $|y_m| = \cosh(r_m)$ and $|z_m| = \sinh(r_m)$ with hyperbolic angles $r_m \coloneqq \operatorname{atanh}|z_m/y_m|$, and phases $e^{-i\phi_m} \coloneqq y_m/|y_m|$ and $e^{i\varphi_m} \coloneqq - e^{i\phi_m}z_m/|z_m|$.
    Then, consider the following unitary
    \begin{equation}
        U \coloneqq \prod_{m=1}^M \exp(-i \phi_m a_m^\dag a_m) \; \exp(\frac{r}{2}\pqty{e^{-i(\phi_m+\varphi_m)}a_m^2-e^{i(\phi_m+\varphi_m)}a_m^{\dag2}}).
    \end{equation}
    which are, from right to left, local squeezing and phase shifts.
    They are Gaussian unitaries whose actions on the annihilation operators are known, from which it can be verified that $U^\dag \vec{a} U = \vec{y}\circ\vec{a} + \vec{z}\circ \vec{a}^{\dagger}$ \cite{WalschaersNonGaussian2021}.
    Further more, the Wigner function of a state transformed by such a $U$ is also known to be $W_{U \rho U^\dag}(\vec{\alpha}) = W_{\rho}(\vec{y}\circ\vec{\alpha} + \vec{z}\circ\vec{\alpha}^*)$ \cite{quasiprobability-review},
    and since the above results are true for all states $\rho$, they must also be true for $U \rho U^\dag$.
    Therefore, they hold for $W_{U \rho U^\dag}(\alpha\vec{1}) = W_{\rho}(\alpha\vec{y}+\alpha^*\vec{z})$, as desired.
\end{proof}

Building upon this lemma, we find a GME inequality concerning the absolute negativity volume along a 2D slice.

\begin{reptheorem}{thm:GME-2D-slice}[Sufficient Wigner negativity volume along a suitably-chosen two-dimensional slice implies GME]
    Choose some coefficients $\vec{y},\vec{z}\in\mathbb{C}^M:\vec{y}\circ\vec{y}^* - \vec{z}\circ\vec{z}^* = \vec{1}$.
    This specifies a two-dimensional slice $\{\alpha\vec{y} + \alpha^*\vec{z} : \alpha \in \mathbb{C}\}$ in phase space.
    Define the negativity volume of the Wigner function along this two-dimensional slice as
    \begin{equation}
        \mathcal{N}_{2D}(\rho)
        \coloneqq \pqty{\frac{\pi}{2}}^{M-1}\!\!
            \int_{\mathbb{C}}
            \dd[2]{\alpha}\left\{\begin{array}{lr}
                0 & \hspace{-6em}\text{if $W_{\rho}(\alpha\vec{y}+\alpha^*\vec{z})\geq0$,}\\[0.5ex]
                \abs\big{W_{\rho}(\alpha\vec{y}+\alpha^*\vec{z})} & \text{otherwise}.
            \end{array}\right.
    \end{equation}
    Then, $\mathcal{N}_{2D}(\rho) > \mathcal{N}_{2D}^{\operatorname{GME}}(\rho)$ implies that $\rho$ is GME, where
    \begin{equation}
        \mathcal{N}_{2D}^{\operatorname{GME}}(\rho)
        \coloneqq \frac{1}{4\sqrt{M-1}} -
        \frac{\pi^{M-1}}{2^M} \!\! \int_{\mathbb{C}}
        \dd[2]{\alpha} W_{\rho}(\alpha \vec{y}+\alpha^* \vec{z}).
    \end{equation}
\end{reptheorem}
\begin{proof}
    Let us first rewrite the integrand as
    \begin{equation}
        \frac{1}{2}\pqty\Big{
            \abs{
                W_{\rho}(\alpha\vec{y}+\alpha^*\vec{z})
            } -
            W_{\rho}(\alpha\vec{y}+\alpha^*\vec{z})
        } = \begin{cases}
                0 & \text{if $W_{\rho}(\alpha\vec{y}+\alpha^*\vec{z})\geq0$,} \\
            \abs{W_{\rho}(\alpha\vec{y}+\alpha^*\vec{z})} & \text{otherwise}.
        \end{cases}
    \end{equation}
    Doing so, we have
    \begin{equation}\label{eq:GME-2D-slice-proof-0}
        \mathcal{N}_{2D}(\rho)
        = \frac{1}{2}\pqty{\frac{\pi}{2}}^{M-1}\int_{
            \mathbb{C}
        } \dd[2]{\alpha} \abs{
            W_{\rho}(\alpha\vec{y}+\alpha^*\vec{z})
        } - \frac{1}{2}\pqty{\frac{\pi}{2}}^{M-1}\int_{
            \mathbb{C}
        } \dd[2]{\alpha}
            W_{\rho}(\alpha\vec{y}+\alpha^*\vec{z}).
    \end{equation}
    Now, take $\rho = \rho_{\mathcal{A}}\otimes\rho_{\bar{\mathcal{A}}}$ separable over the bipartition $\mathcal{A}=\{m_n\}_{n=1}^{\abs{\mathcal{A}}}$ and $\bar{\mathcal{A}}=\{m\}_{m=1}^M \setminus \mathcal{A}$, where $1\leq \abs{\mathcal{A}}<M$ and $\rho_{\{m_1, m_2, \ldots, m_{\abs{\mathcal{A}}}\}}$ are states defined locally on the $\{a_{m_1}, a_{m_2}, \ldots, a_{m_{\abs{\mathcal{A}}}}\}$ modes.
    Let us also write $\vec{y}_{\mathcal{A}} = (y_m)_{m \in \mathcal{A}}$, $\vec{y}_{\bar{\mathcal{A}}} = (y_m)_{m \in \mathcal{A}}$, with $\vec{z}_{\mathcal{A}}$ and $\vec{z}_{\bar{\mathcal{A}}}$ analogously defined.
    Then, the first term is bounded as
    \begin{equation}
    \begin{aligned}
        \rho = \rho_{\mathcal{A}}\otimes\rho_{\bar{\mathcal{A}}} \implies \int_{\mathbb{C}}\dd[2]{\alpha} \abs{W_{\rho}(\alpha\vec{y}+\alpha^*\vec{z})} &= \int_{\mathbb{C}}\dd[2]{\alpha} \abs{
            W_{\rho_{\mathcal{A}}}(\alpha\vec{y}_{\mathcal{A}}+\alpha^*\vec{z}_{\mathcal{A}})
            W_{\rho_{\bar{\mathcal{A}}}}(\alpha\vec{y}_{\bar{\mathcal{A}}}+\alpha^*\vec{z}_{\bar{\mathcal{A}}})
        } \\
        &\leq \sqrt{
            \pqty{\int_{\mathbb{C}}\dd[2]{\alpha} \abs{
                W_{\rho_{\mathcal{A}}}(\alpha\vec{y}_{\mathcal{A}}+\alpha^*\vec{z}_{\mathcal{A}})
            }^2}
            \pqty{\int_{\mathbb{C}}\dd[2]{\alpha} \abs{
                W_{\rho_{\bar{\mathcal{A}}}(\alpha\vec{y}_{\bar{\mathcal{A}}}+\alpha^*\vec{z}_{\bar{\mathcal{A}}})})
            }^2}
        } \\
        &\leq \sqrt{\frac{1}{2\abs{\mathcal{A}}}
        \pqty{\frac{2}{\pi}}^{2\abs{\mathcal{A}}-1}
        \frac{1}{2(M-\abs{\mathcal{A}})}\pqty{\frac{2}{\pi}}^{2(M-\abs{\mathcal{A}})-1}} \\
        &= \frac{1}{2\sqrt{\abs{\mathcal{A}}(M-\abs{\mathcal{A}})}}\pqty{\frac{2}{\pi}}^{M-1},
    \end{aligned}
    \end{equation}
    where we used the Cauchy--Schwarz inequality in the second line and \cref{lem:Wigner-absolute-volume-all} in the penultimate line.
    Therefore, for any $\rho \notin \mathcal{S}_{\operatorname{GME}} \implies \rho = \sum_{(\mathcal{A} \mid \bar{\mathcal{A}})} p_{\mathcal{A}} \rho_{\mathcal{A}} \otimes \rho_{\bar{\mathcal{A}}}$ that is a convex combination over all bipartitions $(\mathcal{A}\mid\bar{\mathcal{A}})$,
    \begin{equation}
    \begin{aligned}
        \rho \notin \mathcal{S}_{\operatorname{GME}} \implies \int_{\mathbb{C}}\dd[2]{\alpha} \abs{W_\rho(\alpha\vec{y}+\alpha^*\vec{z}))} &=
        \int_{\mathbb{C}}\dd[2]{\alpha} \abs{
            \sum_{(\mathcal{A} \mid \bar{\mathcal{A}})} p_{\mathcal{A}}\; W_{\rho_{\mathcal{A}}}(\alpha\vec{y}_{\mathcal{A}}+\alpha^*\vec{z}_{\mathcal{A}})
            W_{\rho_{\bar{\mathcal{A}}}}(\alpha\vec{y}_{\bar{\mathcal{A}}}+\alpha^*\vec{z}_{\bar{\mathcal{A}}})
        } \\
        &\leq
        \sum_{(\mathcal{A} \mid \bar{\mathcal{A}})} p_{\mathcal{A}} \int_{\mathbb{C}}\dd[2]{\alpha} \abs{
            W_{\rho_{\mathcal{A}}}(\alpha\vec{y}_{\mathcal{A}}+\alpha^*\vec{z}_{\mathcal{A}})
            W_{\rho_{\bar{\mathcal{A}}}}(\alpha\vec{y}_{\bar{\mathcal{A}}}+\alpha^*\vec{z}_{\bar{\mathcal{A}}})
        } \\
         &\leq \sum_{(\mathcal{A} \mid \bar{\mathcal{A}})} p_{\mathcal{A}} \; \frac{1}{2\sqrt{\abs{\mathcal{A}}(M-\abs{\mathcal{A}})}}\pqty{\frac{2}{\pi}}^{M-1} \\
         &\leq \max_{1\leq \abs{\mathcal{A}} < M} \frac{1}{2\sqrt{\abs{\mathcal{A}}(M-\abs{\mathcal{A}})}}\pqty{\frac{2}{\pi}}^{M-1}  \sum_{(\mathcal{A} \mid \bar{\mathcal{A}})} p_{\mathcal{A}}  \\
         &= \frac{1}{2\sqrt{M-1}}\pqty{\frac{2}{\pi}}^{M-1}.
     \end{aligned}
    \end{equation}
    Finally, we have
    \begin{equation}
    \begin{aligned}
        \rho \notin \mathcal{S}_{\operatorname{GME}} \implies \mathcal{N}_{2D}(\rho)
        &= \frac{1}{2}\pqty{\frac{\pi}{2}}^{M-1}\int_{
            \mathbb{C}
        } \dd[2]{\alpha} \abs{
            W_{\rho}(\alpha\vec{y}+\alpha^*\vec{z})
        } - \frac{1}{2}\pqty{\frac{\pi}{2}}^{M-1}\int_{
            \mathbb{C}
        } \dd[2]{\alpha}
            W_{\rho}(\alpha\vec{y}+\alpha^*\vec{z}) \\
        &\leq \frac{1}{2}\pqty{\frac{1}{2\sqrt{M-1}}}
        - \frac{1}{2}\pqty{\frac{\pi}{2}}^{M-1}\int_{
            \mathbb{C}
        } \dd[2]{\alpha}
            W_{\rho}(\alpha\vec{y}+\alpha^*\vec{z}),
    \end{aligned}
    \end{equation}
    and taking the contraposition of this statement completes the proof of Theorem~1.
\end{proof}
It is worth noting here that the GME bound is positive and finite.
From \cref{lem:Wigner-absolute-volume-all}, we have
\begin{equation}
    \max_{\rho,\vec{y},\vec{z}}\abs{\mathcal{N}_{2D}^{\operatorname{GME}}(\rho) - \frac{1}{4\sqrt{M-1}}}
    = \frac{1}{2}\pqty{\frac{\pi}{2}}^{M-1} \max_{\rho,\vec{y},\vec{z}}\abs{
        \int_{\mathbb{C}}\dd[2]{\alpha} W_\rho(\alpha\vec{y} + \alpha^*\vec{z})
    }
    = \frac{1}{2M}.
\end{equation}
Since $M \geq 2\sqrt{M-1}$ for all $M \geq 2$, we have
\begin{equation}
    0 \leq \frac{1}{4\sqrt{M-1}} - \frac{1}{2M}
    = \min_{\rho,\vec{y},\vec{z}}\mathcal{N}_{2D}^{\operatorname{GME}}(\rho)
    \leq \mathcal{N}_{2D}^{\operatorname{GME}}(\rho)
    \leq \max_{\rho,\vec{y},\vec{z}}\mathcal{N}_{2D}^{\operatorname{GME}}(\rho)
    = \frac{1}{4\sqrt{M-1}} + \frac{1}{2M}.
\end{equation}

\section{\label{apd:GMN-reduced-state}Proof of Theorem~\ref*{thm:GMN-reduced-state}}
Let us begin by restating a previous result on bipartite entanglement that was proven in Ref.~\cite{ZawCertifiable2024}.
\begin{lemma}[Simplified and rephrased from Theorem 2 of Ref.~\cite{ZawCertifiable2024}]\label{lemma:bipartite-bound}
    Consider a system with two local modes $\{a_1,a_2\}$.
    Define also the collective modes $\sqrt{2}a_{\pm} \coloneqq (y_1 a_{1} + z_1 a_1^\dag) \pm (y_2 a_{2} + z_2 a_2^\dag)$ where $\forall m : |y_m|^2 - |z_m|^2 = 1$, the corresponding collective phase space coordinates $\alpha_{\pm}$, and the partial trace of the state $\rho$ over $a_-$ as $\tr_-\rho$.
    Then, denoting $\rho_m^{(k)}$ as a state defined locally on the $a_m$ mode, and further specifying the probabilities $p_k \geq 0$,
    \begin{equation}
        \rho = \sum_k p_k \rho_{1}^{(k)}\otimes\rho_{2}^{(k)}
        \implies W_{\tr_{-}\rho}(\alpha_+) \geq 0.
    \end{equation}
\end{lemma}

Our next step is to lift the bipartite case into the multipartite scenario, which leads to the following lemma.
\begin{lemma}\label{lemma:multipartite-bound}
    Consider an extended system with $M+(M-2)$ modes $\{a_m\}_{m=1}^{M} \cup \{a_m\}_{m=M+1}^{2M-2}$, where the state of interest $\rho$ is defined on the first $M$ modes while the last $M-2$ modes are auxiliary modes.
    Define the center-of-mass mode of the extended system as $\sqrt{2(M-1)}a_{+} \coloneqq \sum_{m=1}^{2M-2}(y_ma_m+z_ma_m^\dag)$, where $\vec{y}\circ\vec{y}^* - \vec{z}\circ\vec{z}^* = \vec{1}$.
    Define also the relative modes $\{a_{-m}\}_{m=2}^{2M-2}$ such that $\{a_{+}\} \cup \{a_{-m}\}_{m=2}^{2M-2}$ satisfies the canonical commutation relations.
    Then, denoting $\tr_{-} \coloneqq \tr_{a_{-2}}\tr_{a_{-3}}\cdots\tr_{a_{-(2M-2)}}$ as the partial trace over all of the relative modes,
    \begin{equation}
        \rho \notin \mathcal{S}_{\operatorname{GME}} : R = \rho \otimes \bigotimes_{m=M+1}^{2M-2} \varrho_{\{m\}}
        \implies W_{\tr_{-}R}(\alpha) \geq 0,
    \end{equation}
    where $\varrho_{\{m_1, m_2, \ldots\}}$ are states defined locally on the $\{a_{m_1}, a_{m_2}, \ldots\}$ modes and $\rho$ is defined on the first $M$ modes.
\end{lemma}
\begin{proof}
    Using both the convexity of the Wigner function $W_{\sum_k p_k \rho^{(k)}}(\vec{\alpha}) = \sum_k p_k W_{\rho^{(k)}}(\vec{\alpha})$ \cite{quasiprobability-review} and the partial trace, and defining the subnormalised state $\rho_{(\mathcal{A}\mid\bar{\mathcal{A}})} \coloneqq \sum_{k} p_{\mathcal{A}}^{(k)} \rho_{\mathcal{A}}^{(k)} \otimes \rho_{\bar{\mathcal{A}}}^{(k)} \otimes \bigotimes_{m=M+1}^{2M-2} \varrho_{\{m\}}$, we have
    \begin{equation}
    \begin{gathered}
        \rho \notin \mathcal{S}_{\operatorname{GME}} : R = \rho \otimes \bigotimes_{m=M+1}^{2M-2} \varrho_{\{m\}} = \sum_{(\mathcal{A}\mid\bar{\mathcal{A}})}
        \rho_{(\mathcal{A}\mid\bar{\mathcal{A}})} \\
    \begin{aligned}
        \implies W_{\tr_- R}(\alpha) &= \sum_{(\mathcal{A}\mid\bar{\mathcal{A}})}
        W_{
            \tr_-\rho_{(\mathcal{A}\mid\bar{\mathcal{A}})}
        }(\alpha).
    \end{aligned}
    \end{gathered}
    \end{equation}
    Now, for a given $(\mathcal{A}\mid\bar{\mathcal{A}})$, let us partition the modes into the collection $\mathcal{B} \coloneqq \mathcal{A} \cup \{m\}_{m=M+1}^{2M-1-|\mathcal{A}|}$ and $\bar{\mathcal{B}} \coloneqq \{m\}_{m=1}^{2(M-1)} \setminus \mathcal{B}$.
    Here, $(\mathcal{B}\mid\bar{\mathcal{B}})$ is an equal bipartition of the $2(M-1)$ modes since $|\mathcal{B}| = |\bar{\mathcal{B}}| = M-1$.
    This means that
    \begin{equation}
        \rho_{(\mathcal{A}\mid\bar{\mathcal{A}})} = \sum_{k} p_{\mathcal{A}}^{(k)} \underbrace{\pqty{
            \rho_{\mathcal{A}}^{(k)} \otimes
            \bigotimes_{m=M+1}^{2M-1-|\mathcal{A}|} \varrho_{\{m\}}
        }}_{\coloneqq \rho_{\mathcal{B}} } \otimes \underbrace{\pqty{
            \rho_{\bar{\mathcal{A}}}^{(k)} \otimes \bigotimes_{m=2M-|\mathcal{A}|}^{2M-2} \varrho_{\{m\}}
        }}_{\coloneqq\rho_{\bar{\mathcal{B}}}}.
    \end{equation}
    Next, define the collective modes $\{b_m = \sum_{m'\in\mathcal{B}}(u_{m,m'}a_{m'}+v_{m,m'}a_{m'}^\dagger)\}_{m=1}^{M-1}$ on the partition $\mathcal{B}$ such that $\{b_m\}_{m=1}^{M-1}$ satisfies the canonical commutation relations and $\sqrt{M-1} b_1 = \sum_{m\in\mathcal{B}}(y_ma_m + z_ma_m^\dag)$.
    Do the same with $\{\bar{b}_m = \sum_{m'\in\bar{\mathcal{B}}}(\bar{u}_{m,m'}a_{m'} + \bar{v}_{m,m'}a_{m'}\}_{m=1}^{M-1}$ on $\bar{\mathcal{B}}$ such that $\sqrt{M-1} \,\bar{b}_1 = \sum_{m\in\bar{\mathcal{B}}}(y_ma_m + z_ma_m^\dag)$.
    Denote the partial trace over $\{b_m\}_{m=2}^{M-1}$ as $\tr_{-\mathcal{B}}$ and that over $\{\bar{b}_m\}_{m=2}^{M-1}$ as $\tr_{-\bar{\mathcal{B}}}$.
    Then, we have
    \begin{equation}
        \tr_{-\mathcal{B}}\tr_{-\bar{\mathcal{B}}}\rho_{(\mathcal{A}\mid\bar{\mathcal{A}})} = \sum_{k} p_{\mathcal{A}}^{(k)} \tr_{-\mathcal{B}}\rho_{\mathcal{B}}^{(k)} \otimes \tr_{-\bar{\mathcal{B}}}\rho_{\bar{\mathcal{B}}}^{(k)}.
    \end{equation}
    At this stage, we have a separable bipartite state defined on the modes $b_1$ and $\bar{b}_1$, and thus we can directly apply \cref{lemma:bipartite-bound}.
    Define the collective modes $\sqrt{2} b_{\pm} \coloneqq b_1 \pm \bar{b}_1$ on this bipartite system, and the partial trace over $b_-$ as $\tr_{-b}$.
    Then, \cref{lemma:bipartite-bound} implies that
    \begin{equation}
        W_{\tr_{-b}\tr_{-\mathcal{B}}\tr_{-\bar{\mathcal{B}}}\rho_{(\mathcal{A}\mid\bar{\mathcal{A}})}}(\alpha_+) \geq 0.
    \end{equation}
    Now, the single-mode state that remains after $M-1$ partial traces is defined on the mode
    \begin{equation}
        b_+ = \frac{b_{1} + \bar{b}_{1}}{\sqrt{2}}
        = \frac{
            \sum_{m\in\mathcal{B}}(y_ma_m + z_ma_m^\dag) +
            \sum_{m'\in\bar{\mathcal{B}}}(y_{m'}a_{m'} + z_{m'}a_{m'}^\dag)
        }{\sqrt{2(M-1)}}
        = \frac{1}{\sqrt{2(M-1)}}\sum_{m=1}^{2M-2}(y_{m'}a_{m'} + z_{m'}a_{m'}^\dag).
    \end{equation}
    In other words, $b_+ = a_+$ for $a_+$ as defined in the statement of the lemma, and therefore $\tr_{-b}\tr_{-\mathcal{B}}\tr_{-\bar{\mathcal{B}}}\rho_{(\mathcal{A}\mid\bar{\mathcal{A}})} = \tr_-\rho_{(\mathcal{A}\mid\bar{\mathcal{A}})}$ with $\tr_-$ also as defined above.

    Since the preceding steps hold for any choice of bipartition, we have
    \begin{equation}
        \rho \notin \mathcal{S}_{\operatorname{GME}} : R = \rho  \otimes \bigotimes_{m=M+1}^{2M-2} \varrho_{\{m\}} = \sum_{(\mathcal{A}\mid\bar{\mathcal{A}})} \rho_{(\mathcal{A}\mid\bar{\mathcal{A}})}
        \implies W_{\tr_{-}R}(\alpha) = \sum_{(\mathcal{A}\mid\bar{\mathcal{A}})} W_{\tr_-\rho_{(\mathcal{A}\mid\bar{\mathcal{A}})}}(\alpha) \geq 0,
    \end{equation}
    where in the last line, we used the fact that sums of positive functions must be positive, which completes our proof.
\end{proof}

\begin{reptheorem}{thm:GMN-reduced-state}[Negativity of the smoothed Wigner function of the center-of-mass implies GME]
    Choose $M-2$ states $\mathcal{R} = \{\varrho_m\}_{m=1}^{M-2}$. Define the smoothed Wigner function of the center-of-mass of the system as
\begin{equation}
\widetilde{W}_{\tr_{-}\rho}(\alpha;\mathcal{R})\\
    \coloneqq \int_{\mathbb{C}}\dd[2]{\beta}
    \;W_{\tr_{-}\rho}(\beta)
    \;K\pqty{\alpha-\beta;\mathcal{R}},
\end{equation}
where $K(\alpha,\mathcal{R})$ is the convolution kernel
\begin{equation}
\begin{aligned}
    K(\alpha;\mathcal{R})
    &\coloneqq
    2(1-M^{-1})\int_{\mathbb{C}^{M-2}}\dd[2(M-2)]{\vec{\gamma}}
    \;\prod_{m=1}^{M-2}W_{\varrho_m}(\gamma_m)
    \;\delta\pqty{
        \alpha-\tfrac{\vec{\gamma}^T\vec{1}}{\sqrt{M}}
    }.
\end{aligned}
\end{equation}
Then, the smoothed Wigner function lower bounds, up to a factor, the trace distance to all non-GME states as
\begin{equation}
\begin{aligned}
    \max\Bqty{0,-\widetilde{W}_{\tr_{-}\rho}(\alpha;\mathcal{R})} \leq \frac{2}{\pi}\min_{\sigma \notin \mathcal{S}_{\operatorname{GME}}} \| \sigma - \rho \|_1.
\end{aligned}
\end{equation}
Hence, $\exists\alpha : \widetilde{W}_{\tr_{-}\rho}(\alpha;\mathcal{R}) < 0$ implies that $\rho$ is GME.
\end{reptheorem}
\begin{proof}
    We start with the center-of-mass Wigner function of the extended state $R \coloneqq \rho\otimes\bigotimes_{m=M+1}^{2M-2}\varrho_{\{m\}}$ from \cref{lemma:multipartite-bound} by writing it as an expectation value of $\rho$:
    \begin{equation}
    \begin{aligned}
        W_{\tr_{-} R}(\alpha) &= \frac{2}{\pi}\tr(
            \pqty{
                \rho\otimes\bigotimes_{m=M+1}^{2M-2}\varrho_{\{m\}}
            } e^{i\pi\abs{a_+-\alpha}^2}
        ) \\
        &= \frac{2}{\pi}\tr_{\{a_m\}_{m=1}^M}\pqty\Bigg{\rho\; \underbrace{\tr_{\{a_m\}_{m=M+1}^{2M-2}}\pqty{
            \bigotimes_{m=M+1}^{2M-2}\varrho_{\{m\}} e^{i\pi\abs{a_+-\alpha}^2}
        }}_{\eqqcolon V}} \\
        &= \frac{2}{\pi}\tr(\rho V).
    \end{aligned}
    \end{equation}
    Notice that $V$ is Hermitian and that $\|V\| \leq |\langle e^{i\pi\abs{a_+-\alpha}^2}\rangle| \leq 1$, which means that $-\mathbbm{1} \preceq V \preceq \mathbbm{1}$.
    Then, using the variational definition of the trace norm,
    \begin{equation}
        \sigma\notin\mathcal{S}_{\operatorname{GME}} : \|\sigma-\rho\|_1 = \sup_{-\mathbbm{1} \preceq V' \preceq \mathbbm{1}}\tr[(\sigma-\rho)V'] \geq
        \tr[(\sigma-\rho)V] =
        \frac{\pi}{2}W_{\tr_-(\sigma\otimes\bigotimes_m\varrho_m)}(\alpha) - \frac{\pi}{2}W_{\tr_- R}(\alpha) \geq
        - \frac{\pi}{2}W_{\tr_- R}(\alpha).
    \end{equation}
    Therefore, $-W_{\tr_- R}(\alpha) \leq (2/\pi) \min_{\sigma\notin\mathcal{S}_{\operatorname{GME}}}\|\sigma-\rho\|_1$.

    Next, let us reformulate $W_{\tr_- R}$ in terms of a transformation of the Wigner function $W_{\rho}(\vec{\alpha})$ of just the state of interest $\rho$.
    We start by writing $W_{\tr_- R}$ as the partial trace over the full Wigner function over the $2M-2$ modes:
    \begin{equation}
    \begin{aligned}
        W_{\tr_{-} R}(\alpha) &= \int_{\mathbb{C}^{2(M-1)}} \dd[4(M-1)]{\vec{\alpha}}
        \; W_{R}(\vec{\alpha})
        \; \delta\pqty{
            \alpha -
            \tfrac{\vec{y}^T\vec{\alpha} + \vec{z}^T\vec{\alpha}^*}{\sqrt{2(M-1)}}
        } \\
        &=
        \int_{\mathbb{C}^{M}} \dd[2M]{\vec{\alpha}}
            \; W_{\rho}(\vec{\alpha})
        \int_{\mathbb{C}^{M-2}} \dd[2M-4]{\vec{\gamma}}
            \; \prod_{m=1}^{M-2} W_{\varrho_m}(\gamma_m)
        \; \delta\pqty{
            \alpha -
            \sqrt{\tfrac{M}{2(M-1)}}
            \tfrac{\vec{y}_+^T\vec{\alpha} + \vec{z}_+^T\vec{\alpha}^*}{\sqrt{M}} -
            \sqrt{\tfrac{M}{2(M-1)}}
            \tfrac{\vec{y}_-^T\vec{\gamma} + \vec{z}_-^T\vec{\gamma}^*}{\sqrt{M}}
        } \\
        &=
        \int_{\mathbb{C}}\dd[2]{\beta}
        \underbrace{
            \int_{\mathbb{C}^{M}} \dd[2M]{\vec{\alpha}}
            \; W_{\rho}(\vec{\alpha})
            \; \delta\pqty{\beta - \tfrac{\vec{y}_+^T\vec{\alpha} + \vec{z}_+^T\vec{\alpha}^*}{\sqrt{M}}}
        }_{W_{\tr_{-}\rho}(\beta)} \\
        &\qquad\qquad{}\times{}\int_{\mathbb{C}^{M-2}} \dd[2(M-2)]\vec{\gamma}
            \; \prod_{m=1}^{M-2} W_{\varrho_m}(\gamma_m)
        \; \delta\pqty{
            \alpha -
            \sqrt{\tfrac{M}{2(M-1)}}\beta -
            \sqrt{\tfrac{M}{2(M-1)}}
            \tfrac{\vec{y}_-^T\vec{\gamma} + \vec{z}_-^T\vec{\gamma}^*}{\sqrt{M}}
        },
    \end{aligned}
    \end{equation}
    where here $\tr_-\rho$ is the reduced state of $\rho$ in the center-of-mass mode $\propto\sum_{m=1}^M y_m a_m + z_ma_m^*$.
    In the second step, we have also split $\vec{y} \in \mathbb{C}^{2(M-1)}$ into $\vec{y}_+ = (y_m)_{m=1}^{M}$ and $\vec{y}_- = (y_m)_{m=M+1}^{M-2}$, with analogous definitions for $\vec{z}$.
    Now, taking a closer look at the last term, we have
    \begin{equation}
    \begin{aligned}
        &\int_{\mathbb{C}^{M-2}} \dd[2(M-2)]\vec{\gamma}
            \; \prod_{m=1}^{M-2} W_{\varrho_m}(\gamma_m)
        \; \delta\pqty{
            \alpha -
            \sqrt{\tfrac{M}{2(M-1)}}\beta -
            \sqrt{\tfrac{M-2}{2(M-1)}}
            \tfrac{\vec{y}_-^T\vec{\gamma} + \vec{z}_-^T\vec{\gamma}^*}{\sqrt{M-2}}
        } \\
        &\qquad{}={} \frac{2(M-1)}{M}\int_{\mathbb{C}^{M-2}} \dd[2(M-2)]\vec{\gamma}
            \; \prod_{m=1}^{M-2} W_{\varrho_m}(\gamma_m)
        \; \delta\pqty{
            \pqty{
                \sqrt{\tfrac{2(M-1)}{M}}\alpha -
                \beta
            } -
            \tfrac{\vec{y}_-^T\vec{\gamma} + \vec{z}_-^T\vec{\gamma}^*}{\sqrt{M}}
        } \\
        &\qquad{}\eqqcolon{}K\pqty{\sqrt{\tfrac{2(M-1)}{M}}\alpha -
                \beta;\mathcal{R}}.
    \end{aligned}
    \end{equation}
    To simplify the proof, let us note that there always exists a unitary $U_m$ local to the $m$th mode such that $W_{U_m\varrho_m'U_m^\dag}(\gamma_m) = W_{\varrho_m'}(y_m\gamma_m + z_m\gamma_m^*)$, coming from the proof of \cref{lem:Wigner-absolute-volume-all}.
    Using these unitaries, we have
    \begin{equation}
    \begin{aligned}
        &K\pqty{\sqrt{\tfrac{2(M-1)}{M}}\alpha -
                \beta;\mathcal{R}}\\
        &= \frac{2(M-1)}{M}\int_{\mathbb{C}^{M-2}} \dd[2(M-2)]\vec{\gamma}
            \; \prod_{m=1}^{M-2} W_{U_m^\dag \varrho_m U_m}(y_m\gamma_m + z_m\gamma_m^*)
            \; \delta\pqty{
                \pqty{
                    \sqrt{\tfrac{2(M-1)}{M}}\alpha -
                    \beta
                } -
            \tfrac{\vec{y}_-^T\vec{\gamma} + \vec{z}_-^T\vec{\gamma}^*}{\sqrt{M}}
            } \\
        &= \frac{2(M-1)}{M}\int_{\mathbb{C}^{M-2}} \dd[2(M-2)]\vec{\gamma}'
            \; \prod_{m=1}^{M-2} W_{U_m^\dag \varrho_m U_m}(y_m')
            \; \delta\pqty{
                \pqty{
                    \sqrt{\tfrac{2(M-1)}{M}}\alpha -
                    \beta
                } -
            \tfrac{\vec{1}^T\vec{\gamma'}}{\sqrt{M}}
            } \\
        &= K\pqty{\sqrt{\tfrac{2(M-1)}{M}}\alpha -
            \beta;\mathcal{R} = \{U_m^\dag\varrho_m U_m\}}.
    \end{aligned}
    \end{equation}
    That is to say, we can always absorb the coefficients $\vec{y}_-$ and $\vec{z}_-$ into a different choice of states $\mathcal{R}$.
    Finally, substituting this back into the definition of the smoothed Wigner function, we have
    \begin{equation}
        \widetilde{W}_{\tr_-\rho}(\alpha;\mathcal{R})
        \coloneqq \int_{\mathbb{C}}\dd[2]{\beta}
        W_{\tr_-\rho}(\beta)
        K(\alpha-\beta;\mathcal{R})
        = W_{\tr_-R}\pqty{\sqrt{\frac{M}{2(M-1)}}\alpha} \geq -\frac{2}{\pi}\min_{\sigma\notin\mathcal{S}_{\operatorname{GME}}}\|\sigma-\rho\|_1,
    \end{equation}
    as desired.
\end{proof}

For easier computation of $K(\alpha;\mathcal{R})$, note also that
\begin{equation}
\begin{aligned}
    K(\alpha;\mathcal{R}) &= \frac{2(M-1)}{M}\int_{\mathbb{C}^{M-2}} \dd[2(M-2)]\vec{\gamma}
        \; \prod_{m=1}^{M-2} W_{\varrho_m}(\gamma_m)
        \; \delta\pqty{
            \alpha -
            \tfrac{\vec{1}^\dagger\vec{\gamma}}{\sqrt{M}}
        } \\
    &= \frac{2(M-1)}{M-2}\int_{\mathbb{C}^{M-2}} \dd[2(M-2)]\vec{\gamma}
        \; \prod_{m=1}^{M-2} W_{\varrho_m}(\gamma_m)
        \; \delta\pqty{
            \sqrt{\tfrac{M}{M-2}}\alpha -
            \tfrac{\vec{1}^\dagger\vec{\gamma}}{\sqrt{M-2}}
        } \\
    &= \frac{2(M-1)}{M-2}W_{\tr_-(\bigotimes_m\varrho_{m})}\pqty{\sqrt{\tfrac{M}{M-2}}\alpha}.
\end{aligned}
\end{equation}
Indeed, by choosing all the auxiliary states $\varrho_m = \varrho_G$ to be the same Gaussian state gives \cite{WalschaersNonGaussian2021}
\begin{equation}
\begin{aligned}
    K(\alpha;\mathcal{R}_G) = \frac{1-M^{-1}}{\pi\sqrt{\det\Sigma}} e^{-\frac{1}{2}\abs{\Sigma^{-\frac{1}{2}}\spmqty{\Re[\alpha-\alpha_0]\\\Im[\alpha-\alpha_0]}}^2}{}\text{ such that }{}\sqrt{\det\Sigma} \geq \frac{M-2}{4M}.
\end{aligned}
\end{equation}

\section{\label{apd:nonclassicality-depth-reduced-state}Proof of Corollary~\ref*{col:nonclassicality-depth-reduced-state}}
\begin{repcorollary}{col:nonclassicality-depth-reduced-state}[Sufficient nonclassicality depth of the center-of-mass implies GME]
    The nonclassicality depth $\tau_c$ of a state $\rho$ is defined as \cite{nonclassicality-depth}
    \begin{equation}
        \tau_c(\rho) \coloneqq \min\!\Bqty{\tau : \forall \alpha : \frac{1}{\pi\tau}\int_{\mathbb{C}}\dd[2]{\beta}
        \;P_\rho(\beta)
        \;e^{-\frac{\abs{\alpha-\beta}^2}{\tau}} \geq 0 },
    \end{equation}
    where $P_\rho(\alpha)$ is the Glauber $P$ function of $\rho$ such that $\rho = \int_{\mathbb{C}}\dd[2]{\alpha} P_{\rho}(\alpha)\ketbra{\alpha}$, and $\ket{\alpha}$ is the coherent state.
    Then, $\tau_c(\tr_{-}\rho) > 1 - M^{-1}$
    implies that $\rho$ is GME.
\end{repcorollary}
\begin{proof}
    Given $\tau_c(\tr_{-}\rho) > 1 - M^{-1}$, this means by the definition of nonclassicality depth that
    \begin{equation}
    \begin{aligned}
        \exists\alpha : 0 &> \frac{1}{\pi(1-M^{-1})}
        \int_{\mathbb{C}}\dd[2]{\beta}
        \;P_{\tr_-\rho}(\beta)
        \;e^{-\frac{\abs{\alpha-\beta}^2}{1-M^{-1}}} \\
        &=
        \int_{\mathbb{C}}\dd[2]{\beta}
        \;P_{\tr_-\rho}(\beta)
        \;\int_{\mathbb{C}}\dd[2]{\gamma}
        \pqty{
            \frac{2}{\pi}
            e^{-2\abs{\gamma-\beta}^2}
        }
        \pqty{
            \frac{1}{\pi(1/2-M^{-1})}
            e^{-\frac{\abs{\gamma-\alpha}^2}{1/2-M^{-1}}}
        },
    \end{aligned}
    \end{equation}
    where we have used the convolution of two Gaussians to split one Gaussian into two.
    Next, we use that the Wigner function itself can be written as a smoothed Glauber $P$ function as \cite{quasiprobability-review}
    \begin{equation}
        W_\rho(\alpha) = \frac{2}{\pi}\int_{\mathbb{C}}\dd[2]{\gamma} P_\rho(\gamma) e^{-2\abs{\alpha-\gamma}^2},
    \end{equation}
    while we can identify the second Gaussian as
    \begin{equation}
        \frac{1}{\pi(1/2-M^{-1})}
        e^{-\frac{\abs{\gamma-\alpha}^2}{1/2-M^{-1}}} =
        \frac{1}{\pi(1/2-M^{-1})}
        e^{
            -\frac{1}{2}\spmqty{\Re[\gamma-\alpha]\\\Im[\gamma-\alpha]}
            \spmqty{
                \frac{4M}{M-2} & 0\\
                0 & \frac{4M}{M-2}
            }
            \spmqty{\Re[\gamma-\alpha]\\\Im[\gamma-\alpha]}
        } =
        \frac{1}{2\pi\sqrt{\det\Sigma}} e^{
            -\frac{1}{2}
            \abs{
                \Sigma^{-\frac{1}{2}}\spmqty{\Re[\gamma-\alpha]\\\Im[\gamma-\alpha]}
            }^2
        }
    \end{equation}
    with $\Sigma \coloneqq (M-2)(4M)^{-1}\mathbb{1}$ such that $\sqrt{\det\Sigma} = (M-2)(4M)^{-1}$.
    Therefore,
    \begin{equation}
    \begin{aligned}
        \tau_c(\tr_{-}\rho) > 1 - M^{-1} \implies
        \exists\alpha : 0 &>
        \int_{\mathbb{C}}\dd[2]{\gamma}
        \pqty{
            \frac{2}{\pi}
            \int_{\mathbb{C}}\dd[2]{\beta}
            P_{\tr_-\rho}(\beta)
            e^{-2\abs{\gamma-\beta}^2}
        }
        \pqty{
            \frac{1}{\pi(1/2-M^{-1})}
            e^{-\frac{\abs{\gamma-\alpha}^2}{1/2-M^{-1}}}
        } \\
        &=\int_{\mathbb{C}}\dd[2]{\gamma}
        W_{\tr_-\rho}(\gamma)\;
        \frac{1}{2\pi\sqrt{\det\Sigma}} e^{
            -\frac{1}{2}
            \abs{
                \Sigma^{-\frac{1}{2}}\spmqty{\Re[\gamma-\alpha]\\\Im[\gamma-\alpha]}
            }^2
        } \\
        &= \frac{1}{2(1-M^{-1})}
        \int_{\mathbb{C}}\dd[2]{\gamma}
        W_{\tr_-\rho}(\gamma) \;
        K(\alpha-\gamma;\mathcal{R}_G) \\
        &= \frac{1}{2(1-M^{-1})}\widetilde{W}_{\tr_-\rho}(\alpha).
    \end{aligned}
    \end{equation}
    Therefore, $\rho$ is GME by Theorem~2.
\end{proof}

\section{\label{apd:GME-Wigner-witness}Proof of Corollary~\ref*{col:GME-Wigner-witness}}
\begin{repcorollary}{col:GME-Wigner-witness}[GME criterion with Wigner function measurements over a finite region]
    Let the absolute volume of the Wigner function on $\{\alpha\vec{y} + \alpha^*\vec{z} : \alpha \in \omega\}$, where $\vec{y}\circ\vec{y}^*-\vec{z}\circ\vec{z}^* = \vec{1}$ and $\omega \subseteq \mathbb{C}$ is Lebesgue-measurable, be
    \begin{equation}
    \begin{aligned}
        \mathcal{V}_{2D}(\rho;\omega) &\coloneqq \pqty{\frac{\pi}{2}}^{M-1}\int_{\omega}\dd[2]{\alpha} \abs\big{W_{\rho}\pqty{\alpha\vec{y} + \alpha^*\vec{z}}}.
    \end{aligned}
    \end{equation}
    Then, $\mathcal{V}_{2D}(\rho;\omega) > (2\sqrt{M-1})^{-1}$ implies that $\rho$ is GME.
\end{repcorollary}
\begin{proof}
    Given that $\mathcal{V}_{2D}(\rho;\omega) > (2\sqrt{M-1})^{-1}$, we have
    \begin{equation}
        \frac{1}{2\sqrt{M-1}} < \pqty{\frac{\pi}{2}}^{M-1}\int_{\omega}\dd[2]{\alpha} \abs\big{W_{\rho}\pqty{\alpha\vec{y} + \alpha^*\vec{z}}}
        \leq \pqty{\frac{\pi}{2}}^{M-1}\int_{\mathbb{C}}\dd[2]{\alpha} \abs\big{W_{\rho}\pqty{\alpha\vec{y} + \alpha^*\vec{z}}},
    \end{equation}
    simply because $\omega\subseteq\mathbb{C}$ and the integral of a positive function is nondecreasing with the size of the integration region.
    From this, we have
    \begin{equation}
    \begin{aligned}
        \frac{1}{2\sqrt{M-1}}
            &< \pqty{\frac{\pi}{2}}^{M-1}\int_{\mathbb{C}}\dd[2]{\alpha} \abs\big{W_{\rho}\pqty{\alpha\vec{y} + \alpha^*\vec{z}}} \\
        \frac{1}{2\sqrt{M-1}} - \pqty{\frac{\pi}{2}}^{M-1}\int_{\mathbb{C}}\dd[2]{\alpha} W_{\rho}\pqty{\alpha\vec{y} + \alpha^*\vec{z}}
            &< \pqty{\frac{\pi}{2}}^{M-1}\int_{\mathbb{C}}\dd[2]{\alpha} \abs\big{W_{\rho}\pqty{\alpha\vec{y} + \alpha^*\vec{z}}} - \pqty{\frac{\pi}{2}}^{M-1}\int_{\mathbb{C}}\dd[2]{\alpha} W_{\rho}\pqty{\alpha\vec{y} + \alpha^*\vec{z}}\\
        \mathcal{N}_{2D}^{\operatorname{GME}}(\rho) &< \mathcal{N}_{2D}(\rho),
    \end{aligned}
    \end{equation}
    and therefore $\rho \in \mathcal{S}_{\operatorname{GME}}$.
\end{proof}

\section{\label{apd:GME-characteristic-witness}Proof of Corollary~\ref*{col:GME-characteristic-witness}}
\begin{repcorollary}{col:GME-characteristic-witness}[GME criterion with characteristic function measurements over finite points]
    Choose $N$ phase-space points $\Xi = \{\xi_n\}_{n=1}^N$ and coefficients $\vec{y},\vec{z}\in\mathbb{C}^M : \vec{y}\circ\vec{y}^* - \vec{z}\circ\vec{z}^* = \vec{1}$.
    Construct the matrix $\mathbf{C}(\rho;\Xi) \in \mathbb{C}^{N \times N}$ as
    \begin{equation}
        [\mathbf{C}(\rho;\Xi)]_{n,n'}
        \coloneqq \frac{1}{N}\chi_{\rho}\pqty\big{
            (\xi_n-\xi_{n'})\vec{y} +
            (\xi_n^*-\xi_{n'}^*)\vec{z}
        }
        = \frac{1}{N}\chi_{\tr_{-}\rho}\pqty{
            \sqrt{M}(\xi_n-\xi_{n'})
        }.
    \end{equation}
    Next, choose $M-2$ states $\mathcal{R} = \{\varrho_m\}_{m=1}^{M-2}$ and construct $\mathbf{K}(\mathcal{R};\Xi)\in\mathbb{C}^{N\times N}$ by computing
    $[\mathbf{K}(\mathcal{R};\Xi)]_{n,n'} \coloneqq \prod_{m=1}^{M-2} \chi_{\varrho_m}(\xi_n-\xi_{n'})$.
    Then, the magnitude of the most negative eigenvalue of their element-wise product
    \begin{equation}
        \mathcal{N}_C(\rho;\Xi,\mathcal{R}) \coloneqq \max\Bqty{0,-\operatorname{mineig}\bqty{
        \mathbf{C}(\rho;\Xi)\circ\mathbf{K}(\mathcal{R};\Xi)
        }}
    \end{equation}
    lower bounds the trace distance to all non-GME states as
    \begin{equation}
        \mathcal{N}_C(\rho;\Xi,\mathcal{R}) \leq \min_{\sigma \notin\mathrm{GME}} \| \sigma - \rho \|_1.
    \end{equation}
    Hence, $\mathbf{C}(\rho;\Xi)\circ\mathbf{K}(\mathcal{R};\Xi) \not\succeq 0$ implies that $\rho$ is GME.
\end{repcorollary}
\begin{proof}
    Bochner's theorem states that \cite{Bochner-theorem-1,Bochner-theorem-2} the positivity of the Wigner function of a state $W_{\rho}(\vec{\alpha}) \geq 0$ implies that the Bochner matrix constructed out of its characteristic function $\mathbf{C}:[\mathbf{C}]_{n,n'} = \chi_{\rho}(\vec{\xi}_n-\vec{\xi}_{n'})/N$ will be positive semidefinite $\mathbf{C}\succeq 0$.
    Together with Theorem~2, this gives
    \begin{equation}
        \rho\notin\mathcal{S}_{\operatorname{GME}}
        \implies R=\rho \otimes \bigotimes_{m=1}^{M-2}\varrho_m: W_{\tr_-R}(\alpha) \geq 0
        \implies [\overline{\mathbf{C}}]_{n,n'} = \frac{1}{N}\chi_{\tr_-R}(\sqrt{2(M-1)}\xi_n-\sqrt{2(M-1)}\xi_{n'})
        : \overline{\mathbf{C}} \succeq 0.
    \end{equation}
    For later convenience, let us take $\sqrt{2(M-1)} a_+ = \sum_{m=1}^{2M-2} (y_m^* a_m - z_m a_m^\dag)$ to be the center-of-mass of the extended system with $\forall m > M : y_m =1, z_m = 0$.
    Now, noting that the characteristic function can be written as the expectation value $\chi_{\rho}(\vec{\xi}) = \tr(\rho e^{\sum_m(\xi_ma_m^\dag - \xi_m^* a_m)})$, we have
    \begin{equation}
    \begin{aligned}
        \chi_{\tr_-R}(\sqrt{2(M-1)}\xi) &= \tr(R \; e^{\sqrt{2(M-1)}\xi a_+^\dag - \sqrt{2(M-1)}\xi^* a_+}
        ) \\
        &= \tr(\rho e^{\sum_{m=1}^{2M-2}[
            \xi (y_m a_m^\dag - z_m^* a_m) -
            \xi^* (y_m^* a_m - z_m a_m^\dag)
        ]}) \prod_{m'=1}^{M-2} \tr(\varrho_{m'} e^{
            \xi a_{m'+M}^\dag -
            \xi^* a_{m'+M}
        }) \\
        &= \tr(\rho e^{\sum_{m=1}^{2M-2}[
            (\xi y_m + \xi^* z_m) a_m^\dag -
            (\xi y_m + \xi^* z_m)^* a_m
        ]}) \prod_{m'=1}^{M-2} \chi_{\varrho_{m'}}(\xi) \\
        &= \chi_{\rho}(\xi\vec{y} + \xi^*\vec{z})\prod_{m=1}^{M-2} \chi_{\varrho_m}(\xi),
    \end{aligned}
    \end{equation}
    where $\vec{y} = (y_m)_{m=1}^M$ and $\vec{z} = (z_m)_{m=1}^M$.
    Then, we immediately have that for the matrices $\mathbf{C}(\rho;\Xi)$ and $\mathbf{K}(\mathcal{R};\Xi)$ as defined in the corollary,
    \begin{equation}
        \rho \notin \mathcal{S}_{\operatorname{GME}} \implies \mathbf{C}(\rho;\Xi)\circ\mathbf{K}(\mathcal{R};\Xi) = \overline{\mathbf{C}} \succeq 0.
    \end{equation}
    The last step is to notice that for a given $\mathbf{C}(\rho;\Xi)\circ\mathbf{K}(\mathcal{R};\Xi)$, its minimum eigenvalue is its inner product with the corresponding normalized eigenvector $\vec{v} : \vec{v}^\dag\vec{v} = 1$ and $[\mathbf{C}\circ\mathbf{K}]\vec{v} = \operatorname{mineig}[\mathbf{C}\circ\mathbf{K}]\vec{v}$ as
    \begin{equation}
        \operatorname{mineig}\bqty{
            \mathbf{C}(\rho;\Xi)\circ\mathbf{K}(\mathcal{R};\Xi)
        } = \vec{v}^\dagger\bqty{\mathbf{C}(\rho;\Xi)\circ\mathbf{K}(\mathcal{R};\Xi)}\vec{v}
        = \tr\bqty\Bigg{\rho \; \underbrace{\pqty{
                \frac{1}{N}\sum_{n,n'}v_n^* v_{n'}\chi_{\varrho_m}(\xi_n-\xi_{n'})
                \overline{D}(\xi_n-\xi_{n'})
            }}_{\coloneqq V}
        },
    \end{equation}
    where we defined $\overline{D}(\xi) \coloneqq e^{\sum_{m=1}^{2M-2}[
            (\xi y_m + \xi^* z_m) a_m^\dag -
            (\xi y_m + \xi^* z_m)^* a_m
    ]}$ for brevity.
    Here, $V$ can be verified to be Hermitian using $\chi_{\rho}^*(\xi) = \chi_\rho(-\xi)$, and by the inequality of matrix norms,
    \begin{equation}
        \max_{\vec{v}}\abs{
            \vec{v}^\dag[\mathbf{C}\circ\mathbf{K}]\vec{v}
        } \leq \max_{n'}\sum_{n=1}^N\abs{[\mathbf{C}\circ\mathbf{K}]_{n,n'}} = \frac{1}{N}\sum_{n=1}^N\abs{\chi_{\tr_-R}(\xi_n-\xi_{n'})} \leq 1,
    \end{equation}
    where we used that the characteristic function is always bounded by one.
    This means that $\forall \rho':\abs{\tr[\rho' V]} \leq 1 \implies -\mathbbm{1} \preceq V \preceq \mathbbm{1}$.
    Therefore, by the variational definition of the trace norm,
    \begin{equation}
    \begin{aligned}
        \sigma\notin\mathcal{S}_{\operatorname{GME}} : \|\sigma-\rho\|_1 &= \sup_{-\mathbbm{1} \preceq V' \preceq \mathbbm{1}}\tr[(\sigma-\rho)V'] \geq
        \tr[(\sigma-\rho)V] \\
        &= \vec{v}^\dagger\bqty{\mathbf{C}(\sigma;\Xi)\circ\mathbf{K}(\mathcal{R};\Xi)}\vec{v} -
        \vec{v}^\dagger\bqty{\mathbf{C}(\rho;\Xi)\circ\mathbf{K}(\mathcal{R};\Xi)}\vec{v} \\
        &\geq - \operatorname{mineig}\bqty{\mathbf{C}(\rho;\Xi)\circ\mathbf{K}(\mathcal{R};\Xi)},
    \end{aligned}
    \end{equation}
    where we also used that $\sigma\notin\mathcal{S}_{\operatorname{GME}} \implies \mathbf{C}(\sigma)\circ\mathbf{K} \succeq 0 \implies \forall \vec{v} : \vec{v}^\dag [\mathbf{C}(\sigma)\circ\mathbf{K}] \vec{v} \geq 0$, which completes the proof.
\end{proof}

\end{document}